\documentclass[journal,onecolumn,12pt]{IEEEtran}
\usepackage{pslatex}
\usepackage{amsfonts,color,morefloats}
\usepackage{amssymb,amsmath,latexsym,amsthm}
\usepackage{cases}
\usepackage[noadjust]{cite}

\newcommand{\fp}{{\mathbb F}_{p}}
\newcommand{\fq}{{\mathbb F}_{q}}
\newcommand{\ftwo}{{\mathbb F}_{2}}
\newcommand{\fqm}{{\mathbb F}_{q^m}}

\newcommand{\Tr}{{\rm {Tr}}}
\newcommand{\C}{{\mathcal{C}}}

\newcommand{\Supp}{{{\rm Supp}}}

\newcommand{\PG}{{{\rm PG}}}

\newtheorem{thm}{Theorem}
\newtheorem{lem}{Lemma}
\newtheorem{cor}{Corollary}

\newtheorem{prop}{Proposition}

\newtheorem{example}{Example}
\newtheorem{remark}{Remark}

\begin{document}
\title{New constructions of optimal linear codes from simplicial complexes}
\date{\today}
%
\author{Zhao Hu, Yunge Xu, Nian Li, Xiangyong Zeng, Lisha Wang and Xiaohu Tang
\thanks{Z. Hu, Y. Xu, X. Zeng and X. Tang are with the Hubei Key Laboratory of Applied Mathematics, Faculty of Mathematics and Statistics, Hubei University, Wuhan, 430062, China. N. Li and L. Wang are with the Hubei Key Laboratory of Applied Mathematics, School of Cyber Science and Technology, Hubei University, Wuhan, 430062, China. X. Tang is also with the Information Coding $\&$ Transmission Key Lab of Sichuan Province, CSNMT Int. Coop. Res. Centre (MoST), Southwest Jiaotong University, Chengdu, 610031, China. Email: zhao.hu@aliyun.com, xuy@hubu.edu.cn, nian.li@hubu.edu.cn, xiangyongzeng@aliyun.com, wangtaolisha@163.com, xhutang@swjtu.edu.cn}
}
\maketitle

\begin{abstract}
In this paper, we construct a large family of projective linear codes over $\fq$ from the general simplicial complexes of $\fq^m$ via the defining-set construction, which generalizes the results of [IEEE Trans. Inf. Theory 66(11):6762-6773, 2020]. The parameters and weight distribution of this class of codes are completely determined. By using the Griesmer bound, we give a necessary and sufficient condition such that the codes are Griesmer codes and a sufficient condition such that the codes are distance-optimal. For a special case, we also present a necessary and sufficient condition for the codes to be near Griesmer codes. Moreover, by discussing the cases of simplicial complexes with one, two and three maximal elements respectively, the parameters and weight distributions of the codes are given more explicitly, which shows that the codes are at most $2$-weight, $5$-weight and $19$-weight respectively. By studying the optimality of the codes for the three cases in detail, many infinite families of optimal linear codes with few weights over $\fq$ are obtained, including Griesmer codes, near Griesmer codes and distance-optimal codes.
\end{abstract}



\section{Introduction}

Let $\mathbb{F}_{q^m}$ be the finite field with $q^m$ elements and $\mathbb{F}_{q^m}^{*}=\mathbb{F}_{q^m}\backslash\{0\}$, where $q$ is a power of a prime $p$ and $m$ is a positive integer. An $[n, k, d]$ linear code $\mathcal{C}$ over $\fq$ is a $k$-dimensional subspace of $\fq^{n}$ with minimum (Hamming) distance $d$.
An $[n,k,d]$ linear code $\mathcal{C}$ over $\fq$ is called distance-optimal if no $[n,k,d+1]$ code exists (i.e., $\C$ has the largest minimum distance for given $n$ and $k$) and it is called almost distance-optimal if there exists an $[n,k,d+1]$ distance-optimal code. An $[n,k,d]$ linear code $\mathcal{C}$ is called optimal (resp. almost optimal) if its parameters $n$, $k$ and $d$ (resp. $d+1$) meet any bound on linear codes with equality \cite{HWPV}. The Griesmer bound \cite{JHG,GSJS} for an $[n,k,d]$ linear code $\C$ over $\fq$ is given by
\[ n\geq  g(k,d):=\sum_{i=0}^{k-1} \lceil \frac{d}{q^i}\rceil,\]
where $\lceil \cdot \rceil$ denotes the ceiling function. An $[n,k,d]$ linear code $\mathcal{C}$ is called a Griesmer code (resp. near Griesmer code) if its parameters $n$ (resp. $n-1$), $k$ and $d$ achieve the Griesmer bound.
Griesmer codes have been an interesting topic of study for many years due to not only their optimality but also their geometric applications \cite{DING1,DING2}. It's worth noting that Griesmer codes are always distance-optimal, and an $[n,k,d]$ near Griesmer code over $\fq$ with $k>1$ is distance-optimal if $q\mid d$ and almost distance-optimal otherwise \cite{HLZWT}.
In coding theory, it's a fundamental problem to construct (distance-)optimal codes. We refer the readers to \cite{HDWZ,HZHZ,HLZWT,HJKN,HJKWY,JYHLL,LXYQ,LYDCTC,WYLCXF,WYZXYQ} and references therein for recent works on the construction of optimal linear codes.

Let $A_{i}$ denote the number of codewords with Hamming weight $i$ in a code $\mathcal{C}$ of length $n$. The weight enumerator of $\mathcal{C}$ is defined by
$1+A_{1}z+A_{2}z^{2}+\ldots +A_{n}z^{n}$. The sequence $(1, A_{1}, A_{2}, \ldots ,A_{n})$ is called the weight distribution of $\mathcal{C}$.
The weight distribution not only contains crucial information about the error correcting capability of the code, but also allows the computation of the error probability of error detection and correction of a given code \cite{KTKT}.
A code is said to be a $t$-weight code if the number of nonzero $A_{i}$ in the sequence $(A_{1}, A_{2}, \ldots ,A_{n})$ is equal to $t$. Linear codes with few weights have applications in secret sharing schemes \cite{ARJD,CCDY}, authentication codes \cite{DCHT,DW}, association schemes \cite{CAGJ}, strongly regular graphs and some other fields.
Thus the study of weight distribution and the construction of linear codes with few weights have attracted much attention in coding theory, and some nice works have focused on the topics (see, for example, \cite{DING5,DD,DN,MS,MQRT,MAS,MSOFA,TLQZH,ZZNL}).

Recently, constructing optimal or good linear codes from mathematical objects attracts much attention and many attempts have been made in this direction. Some linear codes with good parameters have been constructed by using different mathematical objects such as cryptographic functions \cite{DING4,DING5,DD,HWLZ,MAS,TLQZH,TXFC,ZZNL}, combinatorial $t$-designs \cite{DING5,DCTC}, (almost) difference sets \cite{DING1,HZHZ,ZZZZ}, simplicial complexes (also called down-sets) \cite{CHJY,HJKN,JYHLL,WYLCXF,WYZXYQ}, posets \cite{HJKWY} and maximal arcs \cite{HDWZ}. In these various kinds of mathematical objects, simplicial complexes (which are certain subsets of $\fq^m$ with good algebraic structure) are really useful to construct optimal or good linear codes.
The investigation of constructing linear codes from simplicial complexes, to the best of our knowledge, first appeared in  \cite{CHJY} (in 2018), where the authors Chang and Hyun \cite{CHJY} constructed the first infinite family of binary minimal linear codes violating the Ashikhmin-Barg condition \cite{AAAB} by employing simplicial complexes of $\ftwo^m$ with two maximal elements. In 2020, Hyun et al. \cite{JYHLL} constructed infinite families of optimal binary linear codes from the general simplicial complexes of $\ftwo^m$ via the defining-set approach proposed by Ding and Niederreiter \cite{DN}.
Later, by using simplicial complexes of $\ftwo^m$ with one maximal element, several classes of optimal or good binary linear codes with few weights were derived in \cite{LXSM,WYLY,WYZXYQ} via different construction approaches.
Shortly after, simplicial complexes of $\ftwo^m$ with one and two maximal elements were utilized to construct quaternary optimal linear codes in \cite{WYLCXF,ZXWY} by studying new defining sets of ${\mathbb F}_{4}^m$.
Recently, some researchers also concentrated on linear codes constructed from simplicial complexes of $\fq^m$ with $q>2$. Hyun et al. \cite{HJKN}  first defined the simplicial complexes of $\fp^m$ for an odd prime $p$ in 2019, and after that several classes of optimal $p$-ary few-weight linear codes were constructed in \cite{HJKN,SMLX,WYHJ} by using different simplicial complexes of $\fp^m$ with one maximal element. Later, Pan and Liu \cite{PYLY} defined the simplicial complexes of ${\mathbb F}_{3}^m$ in another way and presented three classes of few-weight ternary codes with good parameters from their defined simplicial complexes of ${\mathbb F}_{3}^m$ with one and two maximal elements.

As can be seen from the previous works, there are only few known results on constructing optimal linear codes over $\fq$ by making use of simplicial complexes of $\fq^m$ with $q>2$ and actually no result associated with the general simplicial complexes of $\fq^m$ except the results in \cite{JYHLL} for the case $q=2$, where ``general" means that there is no restriction on the number of maximal elements of the employed simplicial complexes. Moreover, it's worth noting that in \cite{JYHLL} the weight distributions of the binary linear codes constructed from the general simplicial complexes of $\ftwo^m$ were discussed only for the case of simplicial complexes with two maximal elements. Naturally, one has the following interesting question:\\
\indent \emph{Question}: Can we obtain more infinite families of optimal linear codes over $\fq$ from the general simplicial complexes of $\fq^m$ for any prime power $q$ and determine their weight distributions?

The question above is the major motivation of this paper.
In this paper, we first define the simplicial complexes of $\fq^m$ for any prime power $q$ (see the details in next section) in a different way from the definitions given by \cite{HJKN,PYLY} for a prime $p$ and $q=3$ resepctively. Then we employ the general simplicial complexes of $\fq^m$ to construct projective linear codes $\C$ over $\fq$ via the defining-set construction. With detailed computation on certain exponential sums and by using the principle of inclusion-exclusion, we completely determine the parameters and weight distribution of $\C$. We give a necessary and sufficient condition such that $\C$ is a Griesmer code, which is indeed the Solomon-Stiffler code, and provide an explicit computable criterion for $\C$ to be distance-optimal, which can be easily satisfied (see Remark \ref{remark-thm1-d-optimal}). In addition, for a special case of the general simplicial complexes of $\fq^m$, we give a necessary and sufficient condition such that $\C$ is a near Griesmer code. This shows that many (distance-)optimal codes can be produced from this construction. Moreover, for the three cases of simplicial complexes of $\fq^m$ with one, two and three maximal elements respectively, we give their parameters and weight distributions more explicitly and show that $\C$ is at most $2$-weight, $5$-weight and $19$-weight, respectively. In these three cases, infinite families of optimal linear codes with few weights are produced, including Griesmer codes, near Griesmer codes and distance-optimal codes.

Note that our results extend those of \cite{JYHLL} from $\ftwo$ to $\fq$, and thus it answers the second question in \cite{HLZWT}.
Most notably, the main technique of this paper is quite different from that of all previous works on construction of linear codes from simplicial complexes (including the work \cite{JYHLL}), where the main technique used in the previous works is the so-called generating functions. Furthermore, it's worth noting that by the definition of simplicial complexes of $\fq^m$ and the main technique of this paper, the results in the previous works \cite{CHJY,LXSM,WYLY,WYLCXF,WYZXYQ,ZXWY}, where the simplicial complexes of $\ftwo^m$ with one or two maximal elements were employed, can also be generalized to a general $q$.

The rest of this paper is organized as follows. Section \ref{sec2} introduces the concept of the simplicial complexes of $\fq^m$, recalls the defining-set construction of linear codes and the projective Solomon-Stiffler codes, and presents some useful auxiliary results. In Section \ref{sec3}, we investigate the projective linear codes constructed from the general simplicial complexes of $\fq^m$. Their parameters and weight distributions are completely determined and the optimality of these codes is characterized. Furthermore, in Section \ref{sec4}, we deeply discuss these codes for the cases of the simplicial complexes of $\fq^m$ with one, two and three maximal elements respectively, and consequently many infinite families of optimal linear codes with few weights can be obtained. Section \ref{sec6} makes a comparison of our codes to the known ones in the previous works. Section \ref{sec5} concludes this paper.

\section{Preliminaries} \label{sec2}
In this section, we present some preliminaries which will be used for the subsequent sections.

Here we first recall the relation between the finite field $\fqm$ and the vector space $\fq^m$. Let $\{\alpha_{1},\ldots,\alpha_{m}\}$ be a basis of $\fqm$ over $\fq$ and $\{\beta_{1},\ldots,\beta_{m}\}$ be its dual basis, which implies
\[\Tr_{q}^{q^m}(\alpha_{i}\beta_{j})=\left\{\begin{array}{ll}
0,    &   \mbox{ if } i\ne j;\\
1,   &   \mbox{ if } i = j,
\end{array} \right.\]
where $\Tr_{q}^{q^{m}}(\cdot)$ is the trace function from $\fqm$ to $\fq$.
Then every element $x\in \fqm$  can be uniquely represented as $x=\sum_{i=1}^{m} x_{i}\alpha_{i}$ where $x_{i}=\Tr_{q}^{q^m}(\beta_{i}x) \in \fq$, and $\fq^m$ is isomorphic to $\fqm$ under the mapping
\[ (x_{1},x_{2},\ldots,x_{m})\in \fq^m \mapsto  x=\sum_{i=1}^{m} x_{i}\alpha_{i}\in \fqm.\]
Throughout this paper, we identify $x\in \fqm$ with $(x_{1},x_{2},\ldots,x_{m})\in \fq^m$ with respect to the basis $\{\alpha_{1},\ldots,\alpha_{m}\}$ and hence $\fqm$ is identical to $\fq^m$.

\subsection{Simplicial Complexes of $\fq^m$}
Here we introduce the concept of simplicial complexes of $\fq^m$, where $q$ can be any prime power.
For two vectors $u=(u_{1},u_{2},\ldots,u_{m})$ and $v=(v_{1},v_{2},\ldots,v_{m})$ in $\fq^m$, we say that $u$ covers $v$, denoted $v\preceq u$, if $\Supp(v)\subseteq \Supp(u)$, where $\Supp(u)= \{1 \leq i \leq m : u_{i}\ne 0\}$ is the support of $u$.  A subset $\Delta$ of $\fq^m$ is called a simplicial complex if $u\in \Delta$ and $v\preceq u$ imply $v\in\Delta$.
An element $u$ in $\Delta$ with entries $0$ or $1$ is said to be maximal if there is no element $v\in \Delta$ such that $\Supp(u)$ is a proper subset of $\Supp(v)$. For a simplicial complex $\Delta \subseteq \fq^m$, let $\mathcal{F}=\{F_{1},F_{2},\ldots,F_{h}\}$ be the set of maximal elements of $\Delta$, where $h$ is the number of maximal elements in $\Delta$ and $F_{i}$'s are maximal elements of $\Delta$. Let $A_{i}=\Supp(F_{i})$ for $1\leq i \leq h$, which implies $A_{i}\subseteq [m]:=\{1,2,\ldots, m\}$. Note that $A_{i} \setminus A_{j}\ne \emptyset$ for any $1 \leq i \ne j \leq h$ by the definition. Let $\mathcal{A}=\{A_{1},A_{2},\ldots,A_{h}\}$ be the set of supports of maximal elements of $\Delta$, and $\mathcal{A}$ be called the support of $\Delta$, denoted $\Supp(\Delta)=\mathcal{A}$. Then one can see that a simplicial complex $\Delta$ is uniquely generated by $\mathcal{A}$, denoted $\Delta=\langle \mathcal{A}\rangle$.
Notice that both the set of maximal elements $\mathcal{F}$ and the support $\mathcal{A}$ of $\Delta$ are unique for a fixed simplicial complex $\Delta$. For any set $\mathcal{B}$ consisting of some subsets of $[m]$, we say that a simplicial complex $\Delta$ of $\fq^m$ is generated by $\mathcal{B}$, denoted $\Delta=\langle \mathcal{B} \rangle$, if $\Delta$ is the smallest simplicial complex of $\fq^m$ containing every element in $\fq^m$ with the support $B\in \mathcal{B}$.

Notice that the above definition of simplicial complexes of $\fq^m$ is a generalization of the original definition of simplicial complexes of $\ftwo^m$ \cite{AM,JYHLL}, and it is different from the two definitions presented in \cite{HJKN} for $\fp^m$ and in \cite{PYLY} for ${\mathbb F}_{3}^m$.

From another point of view, one can observe that the simplicial complex $\Delta$ with the support $\mathcal{A}=\{A_{1},A_{2},\ldots,A_{h}\}$ is indeed the union of $h$ vector subspaces of dimensions $|A_{i}|$ of the vector space $\fq^m$, where $|S|$ denotes the cardinality of a set $S$. Specially, $\Delta=\langle \{A\} \rangle$ with exactly one maximal element is an $|A|$-dimensional subspace of $\fq^m$, where $A\subseteq [m]$ is the support of the maximal element.

\subsection{The Defining-Set Construction of Linear Codes}
In 2007, Ding and Niederreiter \cite{DN} introduced a nice and generic way to construct linear codes via trace functions.
Let $D \subset \fqm$ and define
\begin{equation} \label{CD}
\C_D=\{c_{a}=(\Tr_{q}^{q^{m}}(ax))_{x\in D}: a\in \fqm\}.
\end{equation}
Then $\C_{D}$ is a linear code over $\fq$ of length $n:=|D|$. The set $D$ is called the defining set of $\C_D$ and the above construction is accordingly called the defining-set construction.
The defining-set construction is fundamental since every linear code over $\fq$ can be expressed as $\C_D$ for some defining set $D$ (possibly multiset) \cite{DING3}. If the defining set $D$ is well chosen, then $\C_D$ may have optimal or good parameters.
Recently, the defining-set construction of linear codes has attracted a lot of attention and many attempts have been made in this direction to obtain good or optimal linear codes.

It's known that $\C_D$ in \eqref{CD} is equivalent to
\begin{equation*}
\C_\textbf{D}=\{c_\textbf{a}=(\textbf{a} \cdot \textbf{x})_{\textbf{x}\in \textbf{D}}: \textbf{a}\in \fq^m\},
\end{equation*}
where $\textbf{D} \subseteq \fq^m$ is identical to $D \subseteq \fqm$ with respect to a basis of $\fqm$ over $\fq$, see its detailed proof for the binary case in \cite{LGCX}.
Let $\textbf{D}=\{\textbf{d}_{1},\textbf{d}_{2}, \ldots, \textbf{d}_{n} \} \subseteq \fq^m$ and $G$ be an $m \times n$ matrix with $\textbf{d}_{i}$'s as its columns, i.e.,
\[G=[\textbf{d}_{1},\textbf{d}_{2}, \ldots, \textbf{d}_{n}].\]
Then the rows of $G$ generate the linear code $\C_D=\C_\textbf{D}$ \cite{XC}, and $G$ is the generator matrix of $\C_D$ if the rank of $G$ is $m$. This gives the generator matrix of $\C_D$.

\subsection{The Projective Solomon-Stiffler Codes}

The dual code of an $[n,k,d]$ linear code $\C$ over $\fq$ is defined by
$ \C^{\bot}=\{x\in \fq^{n}\,\,|\,\, x \cdot y = 0 \,\, {\rm for \,\, all} \,\, y\in \C \},$
where $x \cdot y$ denotes the Euclidean inner product of $x$ and $y$.
A linear code $\C$ is called projective if its dual code has minimum distance at least $3$. Thus the dual of a projective code has better error correcting capability compared to that of a nonprojective code. By definition, it can be readily verified that the code $\C_{D}$ defined by \eqref{CD} is projective if and only if any two elements of $D$ are linearly independent over $\fq$.

Here we recall the concept of projective space to introduce the projective Solomon-Stiffler code smoothly. The projective space (also called projective geometry) $\PG(m-1, q)$ is the set of all the $1$-dimensional subspaces (called points) of the vector space $\fq^m$. Let $m \geq 2$. A point of $\PG(m-1, q)$ is given in homogeneous coordinates by $\mathbf{x}=(x_{0},x_{1}, \ldots, x_{m-1})\in \fq^m$ where all $x_{i}$'s are not all zero (which implies $\textbf{0}\notin \PG(m-1, q)$); each point has $q-1$ coordinate representations, since $\mathbf{x}$ and $\lambda\mathbf{x}$ yield the same $1$-dimensional subspace of $\fq^m$ for any $\lambda \in \fq^*$. Thus any two points of $\PG(m-1, q)$ are linearly independent over $\fq$. Note that $\fqm^*$ can be expressed as
\[\fqm^* = \fq^*\,\,\overline{\fqm^*}=\{yz: y \in \fq^{*} \mbox{ and } z \in \overline{\fqm^*}\}\]
where $z_{i}/z_{j} \notin \fq^{*}$ for each pair of distinct elements $z_{i}$ and $z_{j}$ in $\overline{\fqm^*}$. This defines the set $\overline{\fqm^*}$. Then the set of all points in $\PG(m-1, q)$ can be identified with the set $\overline{\fqm^*}$ since $\fq^m$ is identical to $\fqm$.
Therefore, the code $\C_{D}$ defined by \eqref{CD} is projective if $D\subseteq \PG(m-1, q)$.
In addition, it should be noted that the projective dimension of a projective subspace in $\PG(m-1,q)$ is one less than that of the corresponding subspace in the vector space $\fq^m$.

The well-known Solomon-Stiffler codes \cite{GSJS} are not only projective codes but also Griesmer codes. The Solomon-Stiffler codes are originally defined as follows: Let $S_{m,q}$ be the set of all distinct points in $\PG(m-1,q)$ and $F=\cup_{i=1}^{h}U_{i}$ be a union of $h$ disjoint projective subspaces of dimensions $u_{i}-1$, where $1\leq u_{1} \leq u_{2} \leq \cdots \leq u_{h} <m$ and at most $q-1$ of the subspaces $U_{i}$ have the same dimension. Then the matrix $G=[S_{m,q} \setminus F]$ whose columns correspond to all points in $S_{m,q} \setminus F$ generates a $[(q^m-1-\sum_{i=1}^{h}(q^{u_{i}}-1))/(q-1),m,q^{m-1}-\sum_{i=1}^{h}q^{u_{i}-1}]$ Griesmer code, called the Solomon-Stiffler code over $\fq$. 
By the definition, one can notice that the Solomon-Stiffler code can also be constructed by the defining-set construction as in \eqref{CD} if $D$ is the set of $S_{m,q} \setminus F$, that is, $D=\PG(m-1,q) \setminus F$. In addition, it's interesting that the matrix $[S_{m,q}]$ generates the well-known Simplex code with parameters $[(q^m-1)/(q-1),m,q^{m-1}]$, which is also a projective Griesmer code.

To the best of our knowledge, the weight distribution of the Solomon-Stiffler codes still remains unknown. We will partially answer this interesting question in this paper.

\subsection{Useful Auxiliary Results}
Let $q$ be a power of a prime $p$ and denote the canonical additive character of $\fq$ by
\[\chi(x)=\zeta_{p}^{\Tr_{p}^{q}(x)},\]
where $\zeta_{p}$ is a primitive complex $p$-th root of unity and $\Tr_{p}^{q}(\cdot)$ is the trace function from $\fq$ to $\fp$.

For an $r$-dimensional $\fq$-subspace $H$ of $\fqm$, the dual of $H$ is defined by
\[H^{\bot}=\{v\in \fqm: \Tr_{q}^{q^m}(uv)=0 \mbox{ for all } u\in H\}.\]
The dual $H^{\bot}$ is an $(m-r)$-dimensional $\fq$-subspace of $\fqm$.  Recall that $\{\alpha_{1},\ldots,\alpha_{m}\}$ is a basis of $\fqm$ over $\fq$ and $\{\beta_{1},\ldots,\beta_{m}\}$ is its dual basis.
If $H$ is an $r$-dimensional $\fq$-subspace of $\fqm$ spanned by the set $\{\alpha_{i_{1}},\ldots,\alpha_{i_{r}}\}$, denoted $H=\mbox{\textbf{span}}\{\alpha_{i_{1}},\ldots,\alpha_{i_{r}}\}$,
then $H^{\bot}=\mbox{\textbf{span}}\{\beta_{i}: i \in [m] \backslash \{i_{1},\ldots,i_{r}\} \}$.

The following lemma regarding the exponential sum on the subspace $H$ of $\fqm$ will be helpful to prove our main result.

\begin{lem}(\cite{YLFL}) \label{lem-expsum-H}
Let $H$ be an $r$-dimensional $\fq$-subspace of $\fqm$. Then for $y\in \fqm$ we have
\begin{eqnarray*}
\sum_{x\in H}\zeta_{p}^{\Tr_{p}^{q^m}(yx)}=\left\{\begin{array}{ll}
q^r,    &   \mbox{ if } y\in H^{\bot};\\
0,   &   \mbox{ otherwise}.\\
\end{array} \right.
\end{eqnarray*}
\end{lem}


Let $0 \leq T<q^{m-1}$ be an integer. Then $T$ can be uniquely written as $T=\sum_{j=0}^{m-2}t_{j}q^{j}$, where $0\leq t_{j} \leq q-1$ is an integer for $0\leq j\leq m-2$. Let $v(T)$ (resp. $u(T)$) denote the smallest (resp. largest) integer in the set $\{0\leq j \leq m-2: t_{j}\ne 0\}$ and $\ell(T)=\sum_{j=0}^{m-1}t_{j}$.
Then for any integer $0 \leq T<q^{m-1}$, we have the following result (see the binary case in \cite[Lemma IV.4]{JYHLL}).

\begin{lem} \label{g(m,d)}
Let $0 \leq T<q^{m-1}$ be an integer. Let notation be defined as above. Then we have
\[\sum_{i=0}^{m-1}\lceil \frac{q^{m-1}-T}{q^i} \rceil = \frac{1}{q-1}(q^m-1-qT+\ell(T))\]
and
\[\sum_{i=0}^{m-1}\lceil \frac{q^{m-1}-T+1}{q^i} \rceil = \frac{1}{q-1}(q^m-1-qT+\ell(T))+v(T)+1.\]
\end{lem}
\begin{proof}
We first write $T$ as $T=\sum_{j=0}^{m-2}t_{j}q^{j}=\sum_{j=v(T)}^{u(T)}t_{j}q^{j}$, where $0\leq t_{j} \leq q-1$ is an integer for $0\leq j\leq m-2$.
A direct computation gives
\begin{align*}
\lceil \frac{q^{m-1}-T}{q^i} \rceil=\left\{\begin{array}{ll}
q^{m-i-1}-\sum_{j=v(T)}^{u(T)}t_{j}q^{j-i},& \mbox{if $0\leq i \leq v(T)$};\\
q^{m-i-1}-\sum_{j=i}^{u(T)}t_{j}q^{j-i}+\lceil \frac{-\sum_{j=v(T)}^{i-1}t_{j}q^{j}}{q^i} \rceil,& \mbox{if $v(T) < i \leq u(T)$};\\
q^{m-i-1}+\lceil \frac{-\sum_{j=v(T)}^{u(T)}t_{j}q^{j}}{q^i} \rceil,& \mbox{if $u(T) < i \leq m-1$}.
\end{array} \right.
\end{align*}
Notice that $\sum_{j=v(T)}^{i-1}t_{j}q^{j}\leq (q-1)(q^{v(T)}+\cdots+q^{i-1})=q^i-q^{v(T)}$ for $v(T) < i \leq u(T)$ and
$\sum_{j=v(T)}^{u(T)}t_{j}q^{j} \leq (q-1)(q^{v(T)}+\cdots+q^{u(T)})=q^{u(T)+1}-q^{v(T)}$ for $u(T) < i \leq m-1$, which implies that $\lceil \frac{-\sum_{j=v(T)}^{i-1}t_{j}q^{j}}{q^i} \rceil=0$ for $v(T) < i \leq u(T)$ and $\lceil \frac{-\sum_{j=v(T)}^{u(T)}t_{j}q^{j}}{q^i} \rceil=0$ for $u(T) < i \leq m-1$ respectively.
Thus we have
\begin{align*}
\sum_{i=0}^{m-1}\lceil \frac{q^{m-1}-T}{q^i} \rceil &=\sum_{i=0}^{m-1}q^{m-i-1}-\sum_{i=0}^{v(T)}\sum_{j=v(T)}^{u(T)}t_{j}q^{j-i}
-\sum_{i=v(T)+1}^{u(T)}\sum_{j=i}^{u(T)}t_{j}q^{j-i}\\
&=\frac{q^m-1}{q-1}-\sum_{i=0}^{u(T)}\sum_{j=i}^{u(T)}t_{j}q^{j-i}
=\frac{q^m-1}{q-1}-\sum_{j=0}^{u(T)}t_{j}\sum_{i=0}^{j}q^{i}\\
&=\frac{q^m-1}{q-1}-\sum_{j=0}^{u(T)}t_{j}\frac{q^{j+1}-1}{q-1}=\frac{1}{q-1}(q^m-1-qT+\ell(T)).
\end{align*}

The second assertion can be similarly proved and thus we omit the details of its proof here.
 This completes the proof.
\end{proof}


\section{The projective linear codes over $\fq$ from the general simplicial complexes} \label{sec3}
Let $\Delta$ be a simplicial complex of $\fqm$, and $\Delta^{c}$ be the complement of $\Delta$, namely, $\Delta^{c}=\fq^m \backslash \Delta$. Notice that if $x\in \Delta$, then $yx\in \Delta$ for any $y\in \fq^*$ due to the definition of simplicial complexes. Hence for any simplicial complex $\Delta$, $\Delta^{c}$ can be expressed as
\[\Delta^{c}=\fq^{*}\overline{\Delta}^{c}=\{yz: y \in \fq^{*} \mbox{ and } z \in \overline{\Delta}^{c}\}\]
where $z_{i}/z_{j} \notin \fq^{*}$ for distinct elements $z_{i}$ and $z_{j}$ in $\overline{\Delta}^{c}$, and clearly $|\overline{\Delta}^{c}|=|\Delta^{c}|/(q-1)$. This defines $\overline{\Delta}^{c}$ and $\overline{\Delta}^{c}$ can be viewed as a subset of $\PG(m-1,q)$.

In this section, we investigate the projective codes $\C_{\overline{\Delta}^{c}}$ defined as in \eqref{CD}. Even though  $\C_{\overline{\Delta}^{c}}$ can also be extended to the nonprojective codes $\C_{\Delta^{c}}$, we restrict our discussion to the projective case in this paper.

\begin{thm} \label{optimalcode-anti-sc-pro}
Let $\Delta$ be a simplicial complex of $\fqm$ with the support $\mathcal{A}=\{A_{1},A_{2},\ldots,A_{h}\}$, where $1 \leq |A_{1}| \leq |A_{2}| \leq \cdots \leq |A_{h}|<m$. Assume that $A_{i}\backslash (\cup_{1\leq j \leq h, j\ne i}A_{j}) \ne \emptyset$ for any $1 \leq i \leq h$ and
$q^{m} > \sum_{1\leq i\leq h} q^{|A_{i}|}$. Denote $T=\sum_{1\leq i\leq h} q^{|A_{i}|-1}$. Let $\C_{\overline{\Delta}^{c}}$ be defined as in \eqref{CD}. Then
\begin{enumerate}
\item [1)] $\C_{\overline{\Delta}^{c}}$ has parameters $[(q^m-|\Delta|)/(q-1),m,q^{m-1}-T]$, where $|\Delta|=\sum_{\emptyset \not=S\subseteq \mathcal{A}}(-1)^{|S|-1}q^{|\cap S|}$ and $\cap S$ is defined as $\cap S=\cap_{A\in S}A$.
\item [2)] $\C_{\overline{\Delta}^{c}}$ is a Griesmer code if and only if $|A_{i}\cap A_{j}|=0$ for any $1\leq i<j \leq h$ and at most $q-1$ of $|A_{i}|$'s are the same.
\item [3)] $\C_{\overline{\Delta}^{c}}$ is distance-optimal if $|\Delta|-1+(q-1)(v(T)+1)>qT-\ell(T)$.
\item [4)] $\C_{\overline{\Delta}^{c}}$  has the following weight enumerator 
      \[\sum_{\emptyset \not=R\subseteq \Omega} |\Psi_{R}|
        z^{q^{m-1}-\sum_{S\in R}(-1)^{|S|-1}q^{|\cap S|-1}}+(q^{m-|\cup_{i=1}^{h}A_{i}|}-1)z^{q^{m-1}}+1\]
      where $\Omega=\{S: S \subseteq \mathcal{A}, S\ne \emptyset \}$ and
      \[|\Psi_{R}|=q^{m-|\cup_{S\in \Omega \setminus R}(\cap S)|}-\sum_{\emptyset \not=E\subseteq R}(-1)^{|E|-1}q^{m-|(\cup_{L\in E}(\cap L)) \cup (\cup_{S\in \Omega \setminus R}(\cap S))|}.\]
\end{enumerate}
\end{thm}

\begin{proof}
1).
Let $\Delta_{A_{i}}:=\langle \{A_{i}\}\rangle$ be the simplicial complex of $\fq^m$ with exactly one maximal element for $1\leq i \leq h$.
Clearly, one has $\Delta=\cup_{i=1}^{h}\Delta_{A_{i}}$.
Notice that $\Delta_{A_{i}}$ is indeed a $|A_{i}|$-dimensional $\fq$-subspace of $\fq^m$ for any $1\leq i \leq h$, which implies $|\Delta_{A_{i}}|=q^{|A_{i}|}$. Let $S=\{A_{i}: i \in J\}$ be a nonempty subset of $\mathcal{A}$ for  $J\subseteq [h]:=\{1,2,\ldots,h\}$ and $J\ne \emptyset$. Then $\cap_{i\in J}\Delta_{A_{i}}=\Delta_{\cap S}$, where $\cap S= \cap_{A\in S}A \subseteq [m]$ and $\Delta_{\cap S}$ denotes the simplicial complex with exactly one maximal element generated by $\{\cap S\}$. This implies that $|\cap_{i\in J}\Delta_{A_{i}}|=|\Delta_{\cap S}|=q^{|\cap S|}$.

Using the principle of inclusion-exclusion, it follows that
\begin{eqnarray}\label{Delta-card}
|\Delta|=|\cup_{i=1}^{h}\Delta_{A_{i}}|=\sum_{1\leq i\leq h} |\Delta_{A_{i}}|
-\sum_{1\leq i<j\leq h}|\Delta_{A_{i}}\cap \Delta_{A_{j}}|
+\cdots+(-1)^{h-1} |\cap_{i=1}^{h}\Delta_{A_{i}}|
=\sum_{\emptyset \not=S\subseteq \mathcal{A}}(-1)^{|S|-1}q^{|\cap S|}.
\end{eqnarray}
Then the length of $\C_{\overline{\Delta}^{c}}$ is $n=|\overline{\Delta}^{c}|=|\Delta^{c}|/(q-1)=(q^m-|\Delta|)/(q-1)$.

Now we compute the possible Hamming weights $wt(c_{a})$ of the codewords $c_{a}$ in $\C_{\overline{\Delta}^{c}}$. If $a=0$, we immediately have $wt(c_{a})=0$. For $a\in \fqm^{*}$, one has $wt(c_{a})=n-N_{a}$, where $N_{a}=|\{x\in \overline{\Delta}^{c}: \Tr_{q}^{q^{m}}(ax) = 0\}|=\frac{1}{q-1}|\{x\in \Delta^{c}: \Tr_{q}^{q^{m}}(ax) = 0\}|$.
Using the orthogonal property of nontrivial additive characters, for $a\in \fqm^*$, it gives
\begin{align*}
N_{a}=&\frac{1}{q(q-1)}\sum_{x\in \Delta^{c}}
\sum_{u\in \fq} \chi(u\Tr_{q}^{q^{m}}(ax))
=\frac{1}{q-1}(q^{m-1}-\frac{1}{q}\sum_{x\in \Delta}\sum_{u\in \fq} \chi(u\Tr_{q}^{q^{m}}(ax)))
=\frac{1}{q-1}(q^{m-1}-\Theta),
\end{align*}
where $\Theta:=\frac{1}{q}\sum_{x\in \Delta} \sum_{u\in \fq} \chi(u\Tr_{q}^{q^{m}}(ax))$. Then for $a\in \fqm^{*}$ it leads to
\begin{equation} \label{eq-thm1-wt(a)}
wt(c_{a})=q^{m-1}-(|\Delta|-\Theta)/(q-1).
\end{equation}

To determine the minimum distance of $\C_{\overline{\Delta}^{c}}$, it is equivalent to determining the maximum value of $|\Delta|-\Theta$ as $a$ runs through $\fqm^{*}$. For a simplicial complex $\Delta_{A}=\langle \{A\}\rangle$ with exactly one maximal element, where $A\subseteq [m]$, we define
\[\Phi(\Delta_{A}):= \frac{1}{q}\sum_{u\in \fq} \sum_{x\in \Delta_{A}}\chi(u\Tr_{q}^{q^{m}}(ax)).\]
Since $\Delta_{A}$ can be viewed as a $|A|$-dimensional $\fq$-subspace of $\fqm$, it follows from Lemma \ref{lem-expsum-H} that
\begin{eqnarray} \label{eq-thm1-Phi}
\Phi(\Delta_{A})=\left\{\begin{array}{ll}
q^{|A|},     &   \mbox{ if } a\in \Delta_{A}^{\bot};\\
q^{|A|-1},  &   \mbox{ otherwise},
\end{array} \right.
\end{eqnarray}
where $\Delta_{A}^{\bot}$ is the dual of $\Delta_{A}$. In particular, $\Delta_{A}=\{0\} \subseteq \fqm$ and $\Phi(\Delta_{A})=1$ if $A=\emptyset$.
Again using the principle of inclusion-exclusion gives
\begin{equation*}
\Theta =\sum_{1\leq i\leq h} \Phi(\Delta_{A_{i}})
-\sum_{1\leq i<j\leq h}\Phi(\Delta_{A_{i}}\cap \Delta_{A_{j}})
+\cdots+(-1)^{h-1} \Phi(\cap_{i=1}^{h}\Delta_{A_{i}})
=\sum_{\emptyset \not=S\subseteq \mathcal{A}}(-1)^{|S|-1}\Phi(\Delta_{\cap S}),
\end{equation*}
where $\Delta_{\cap S}=\langle \{\cap S\} \rangle$. This together with \eqref{Delta-card} leads to
\begin{equation} \label{eq-thm1-DeltaOmega-1}
|\Delta|-\Theta=\sum_{\emptyset \not=S\subseteq \mathcal{A}}(-1)^{|S|-1}(q^{|\cap S|}-\Phi(\Delta_{\cap S}))
=\sum_{\emptyset \not=S\subseteq \mathcal{A}}(-1)^{|S|-1}f_{a}(S),
\end{equation}
where $f_{a}(S):=q^{|\cap S|}-\Phi(\Delta_{\cap S})$. Next we express $|\Delta|-\Theta$ in another way.
Define $\mathcal{A}_{i}:=\{A_{1}\cap A_{i},A_{2}\cap A_{i},\ldots, A_{i-1}\cap A_{i}\}$ for $2\leq i \leq h$.
From \eqref{eq-thm1-DeltaOmega-1}, one can further derive
\begin{equation} \label{eq-thm1-DeltaOmega-2}
|\Delta|-\Theta=\sum_{\emptyset \not=S\subseteq \mathcal{A}}(-1)^{|S|-1}f_{a}(S)
=\sum_{i=1}^{h}f_{a}(\{A_{i}\})-\sum_{i=2}^{h}\sum_{\emptyset \not=S\subseteq \mathcal{A}_{i}}(-1)^{|S|-1}f_{a}(S),
\end{equation}
where the second equality can be verified by using the mathematical induction as follows:\\
\indent $(a).$ For $h=2$, it can be readily verified that
\[\sum_{\emptyset \not=S\subseteq \{A_{1},A_{2}\}}(-1)^{|S|-1}f_{a}(S)=\sum_{i=1}^{2}f_{a}(\{A_{i}\})-\sum_{\emptyset \not=S\subseteq \mathcal{A}_{2}}(-1)^{|S|-1}f_{a}(S).\]
\indent $(b).$ Suppose that \eqref{eq-thm1-DeltaOmega-2} holds for $h=s$ where $s \geq 2$ is an integer. Then we have
\begin{align*}
&\sum_{\emptyset \not=S\subseteq \{A_{1},A_{2},\ldots,A_{s+1}\}}(-1)^{|S|-1}f_{a}(S)\\
=&\sum_{\emptyset \not=S\subseteq \{A_{1},A_{2},\ldots,A_{s}\}}(-1)^{|S|-1}f_{a}(S)+f_{a}(\{A_{s+1}\})-\sum_{\emptyset \not=S\subseteq \mathcal{A}_{s+1}}(-1)^{|S|-1}f_{a}(S)\\
=&\sum_{i=1}^{s+1}f_{a}(\{A_{i}\})-\sum_{i=2}^{s+1}\sum_{\emptyset \not=S\subseteq \mathcal{A}_{i}}(-1)^{|S|-1}f_{a}(S),
\end{align*}
which implies that \eqref{eq-thm1-DeltaOmega-2} also holds for $h=s+1$. This proves \eqref{eq-thm1-DeltaOmega-2}.

Now we claim that
\begin{equation} \label{eq-thm1-Fi}
\sum_{\emptyset \not=S\subseteq \mathcal{A}_{i}}(-1)^{|S|-1}f_{a}(S)\geq0
\end{equation}
for any $2\leq i \leq h$.
Let the linear code $\C_{\Delta_{\mathcal{A}_{i}}}$ be defined by $\C_{\Delta_{\mathcal{A}_{i}}}=\{\tilde{c}_{a}=(\Tr_{q}^{q^{m}}(ax))_{x\in \Delta_{\mathcal{A}_{i}}}: a\in \fqm\}$ where $\Delta_{\mathcal{A}_{i}}:=\langle\mathcal{A}_{i}\rangle$ is the simplicial complex generated by the set $\mathcal{A}_{i}$.
For $a\in \fqm^*$, the Hamming weight of the codeword $\tilde{c}_{a}$ in $\C_{\Delta_{\mathcal{A}_{i}}}$ is \[wt(\tilde{c}_{a})=|\Delta_{\mathcal{A}_{i}}|-|\{x\in \Delta_{\mathcal{A}_{i}}: \Tr_{q}^{q^{m}}(ax) = 0 \}|=|\Delta_{\mathcal{A}_{i}}|-\frac{1}{q}\sum_{x\in \Delta_{\mathcal{A}_{i}}}\sum_{u\in \fq} \chi(u\Tr_{q}^{q^{m}}(ax)).\]
Similar to the computation of \eqref{eq-thm1-DeltaOmega-1}, we have
\[wt(\tilde{c}_{a})=\sum_{\emptyset \not=S\subseteq \mathcal{A}_{i}}(-1)^{|S|-1}f_{a}(S)\geq 0.\]
This proves that \eqref{eq-thm1-Fi} holds.

Therefore, by \eqref{eq-thm1-DeltaOmega-2}, \eqref{eq-thm1-Fi} and \eqref{eq-thm1-Phi}, we conclude that
\begin{equation} \label{eq-thm1-DeltaOmega-3}
|\Delta|-\Theta \leq \sum_{i=1}^{h}f_{a}(\{A_{i}\}) \leq (q-1)\sum_{1\leq i\leq h} q^{|A_{i}|-1}.
\end{equation}

Next we show that there exists some $a\in \fqm^*$ such that $|\Delta|-\Theta=(q-1)\sum_{1\leq i\leq h} q^{|A_{i}|-1}$ under the assumption that $A_{i}\setminus (\cup_{1\leq j\leq h, j\ne i}A_{j})\ne \emptyset$ for any $1 \leq i \leq h$.
Notice that by \eqref{eq-thm1-DeltaOmega-1}, $|\Delta|-\Theta=\sum_{i=1}^{h}f_{a}(\{A_{i}\})$ if $\Phi(\Delta_{\cap S})=q^{|\cap S|}$ for any $S\subseteq \mathcal{A}$ with $|S|\geq 2$, that is, $a\in \Delta_{\cap S}^{\bot}$ for any $S\subseteq \mathcal{A}$ with $|S|\geq 2$ due to \eqref{eq-thm1-Phi}. Moreover, by \eqref{eq-thm1-Phi}, $\sum_{i=1}^{h}f_{a}(\{A_{i}\}) = (q-1)\sum_{1\leq i\leq h} q^{|A_{i}|-1}$ if $a\notin \Delta_{A_{i}}^{\bot}$ for any $1\leq i \leq h$. Therefore, one has  $|\Delta|-\Theta=(q-1)\sum_{1\leq i\leq h} q^{|A_{i}|-1}$ if
$a\notin \Delta_{A_{i}}^{\bot}$ for any $1\leq i \leq h$ and $a\in \Delta_{\cap S}^{\bot}$ for any $S\subseteq \mathcal{A}$ with $|S|\geq 2$.
Recall that $\{\alpha_{1},\ldots,\alpha_{m}\}$ is a basis of $\fqm$ over $\fq$ and $\{\beta_{1},\ldots,\beta_{m}\}$ is its dual basis.
Then the subspaces $\Delta_{A_{i}}$ and $\Delta_{\cap S}$ of $\fqm$ can be expressed as $\Delta_{A_{i}}=\mbox{\textbf{span}}\{\alpha_{j}: j \in A_{i}\}$ and
$\Delta_{\cap S}=\mbox{\textbf{span}}\{\alpha_{j}: j \in \cap S\}$ respectively.
Thus $\Delta_{A_{i}}^{\bot}=\mbox{\textbf{span}}\{\beta_{j}: j \in [m] \backslash A_{i}\}$ and
$\Delta_{\cap S}^{\bot}=\mbox{\textbf{span}}\{\beta_{j}: j \in [m] \setminus \cap S\}$.
For any $a\in \fqm$, $a$ can be uniquely represented as $a=\sum_{j=1}^{m}a_{j}\beta_{j}$, where $a_{j}\in \fq$.
Thus $a\notin \Delta_{A_{i}}^{\bot}$ if and only if $a_{j}\ne 0$ for some $j\in A_{i}$, and
$a\in \Delta_{\cap S}^{\bot}$ if and only if $a_{j} = 0$ for any $j\in \cap S$.
Suppose that $a$ is an element of $\fqm$ such that $a_{j}\ne 0$ for some $j\in A_{i}\setminus (\cup_{1\leq j \leq h, j\ne i}A_{j})$ as $i$ runs through $[h]$ and $a_{j}=0$ for other $j$'s in $[m]$.
Since $A_{i}\setminus (\cup_{1\leq j\leq h, j\ne i}A_{j})\ne \emptyset$ for any $1 \leq i \leq h$, such $a\in \fqm^*$ must exist. It can be verified that for this given $a$, one has $a\notin \Delta_{A_{i}}^{\bot}$ for any $1\leq i \leq h$ and $a\in \Delta_{\cap S}^{\bot}$ for any $S\subseteq \mathcal{A}$ with $|S|\geq 2$. Thus we have $|\Delta|-\Theta=(q-1)\sum_{1\leq i\leq h} q^{|A_{i}|-1}$ for the supposed $a\in \fqm^*$. This achieves the desired result.


With above detailed computation on $|\Delta|-\Theta$, we conclude that the maximum value of $|\Delta|-\Theta$ as $a$ runs through $\fqm^{*}$ is $(q-1)\sum_{1\leq i\leq h} q^{|A_{i}|-1}$ by \eqref{eq-thm1-DeltaOmega-3}.
Therefore, it follows from \eqref{eq-thm1-wt(a)} that the minimum distance of $\C_{\overline{\Delta}^{c}}$ is $q^{m-1}-\sum\nolimits_{1\leq i\leq h} q^{|A_{i}|-1}$, which is greater than $0$ since $q^{m} > \sum_{1\leq i\leq h} q^{|A_{i}|}$. This shows that the dimension of $\C_{\overline{\Delta}^{c}}$ is $m$. Then the parameters of $\C_{\Delta^{c}}$ are obtained.

2). Denote the minimum distance of $\C_{\overline{\Delta}^{c}}$ by $d$, i.e, $d=q^{m-1}-T$, where $T=\sum_{1\leq i\leq h} q^{|A_{i}|-1}$. Using Lemma \ref{g(m,d)}, we have
\begin{equation} \label{eq-thm1-Griesmer-d}
g(m,d)=\sum_{i=0}^{m-1}\lceil \frac{q^{m-1}-T}{q^i} \rceil = \frac{1}{q-1}(q^m-1-qT+\ell(T))
\end{equation}
and
\begin{equation}\label{eq-thm1-Griesmer-d+1}
g(m,d+1)=\sum_{i=0}^{m-1}\lceil \frac{q^{m-1}-T+1}{q^i} \rceil = \frac{1}{q-1}(q^m-1-qT+\ell(T))+v(T)+1
\end{equation}
where $\ell(T)$ and $v(T)$ are defined as before.

Recall that the length of $\C_{\overline{\Delta}^{c}}$ is $n=(q^m-|\Delta|)/(q-1)$. Now we prove that $n=g(m,d)$ if and only if $|A_{i}\cap A_{j}|=0$ for any $1\leq i<j \leq h$ and at most $q-1$ of $|A_{i}|$'s are the same.
Similar to the computation of \eqref{eq-thm1-DeltaOmega-2}, $|\Delta|$ in \eqref{Delta-card} can be represented as
\begin{equation}\label{eq-thm1-Delta-card-2}
|\Delta|=\sum_{\emptyset \not=S\subseteq \mathcal{A}}(-1)^{|S|-1}q^{|\cap S|}
=\sum_{i=1}^{h}q^{|A_{i}|}-\sum_{i=2}^{h}\sum_{\emptyset \not=S\subseteq \mathcal{A}_{i}}(-1)^{|S|-1}q^{|\cap S|}
=qT-\sum_{i=2}^{h}|\Delta_{\mathcal{A}_{i}}|,
\end{equation}
where $\Delta_{\mathcal{A}_{i}}=\langle\mathcal{A}_{i}\rangle$ and $\mathcal{A}_{i}$ is defined as before. Then $n-g(m,d)=\frac{1}{q-1}(qT+1-|\Delta|-\ell(T))=\frac{1}{q-1}(\sum_{i=2}^{h}|\Delta_{\mathcal{A}_{i}}|-\ell(T)+1)$ by \eqref{eq-thm1-Griesmer-d} and \eqref{eq-thm1-Delta-card-2}. Thus $n=g(m,d)$ if and only if $\sum_{i=2}^{h}|\Delta_{\mathcal{A}_{i}}|=\ell(T)-1$. Note that $|\Delta_{\mathcal{A}_{i}}|\geq 1$ since $0$ is always in
$\Delta_{\mathcal{A}_{i}}$, and $\ell(T)\leq h$ by the definition of $\ell(T)$. Thus $\sum_{i=2}^{h}|\Delta_{\mathcal{A}_{i}}|=\ell(T)-1$ if and only if $|\Delta_{\mathcal{A}_{i}}|=1$ for $2\leq i\leq h$ and $\ell(T)= h$, which are equivalent to $|A_{i}\cap A_{j}|=0$ for any $1\leq i<j \leq h$ and at most $q-1$ of $|A_{i}|$'s are the same. Therefore, $\C_{\overline{\Delta}^{c}}$ is a Griesmer code if and only if $|A_{i}\cap A_{j}|=0$ for any $1\leq i<j \leq h$ and at most $q-1$ of $|A_{i}|$'s are the same.

3). By \eqref{eq-thm1-Griesmer-d+1}, $\C_{\overline{\Delta}^{c}}$ is distance-optimal if $g(m,d+1)>n$, i.e., $|\Delta|-1+(q-1)(v(T)+1)>qT-\ell(T)$.

4). It follows from \eqref{eq-thm1-wt(a)} and \eqref{eq-thm1-DeltaOmega-1} that the Hamming weight of $c_{a}$ in $\C_{\overline{\Delta}^{c}}$ for $a\in \fqm^*$ is
\begin{equation}\label{eq-thm1-wt(a)-2}
wt(c_{a})=q^{m-1}-\frac{1}{q-1}\sum_{\emptyset \not=S\subseteq \mathcal{A}}(-1)^{|S|-1}(q^{|\cap S|}-\Phi(\Delta_{\cap S})).
\end{equation}
By \eqref{eq-thm1-Phi}, it leads to
\[q^{|\cap S|}-\Phi(\Delta_{\cap S})=\left\{\begin{array}{ll}
0,     &   \mbox{ if } a\in \Delta_{\cap S}^{\bot};\\
(q-1)q^{|\cap S|-1},  &   \mbox{ if } a\notin \Delta_{\cap S}^{\bot}.
\end{array} \right.\]

Denote the set of all nonempty subsets of $\mathcal{A}$ by $\Omega$, that is, $\Omega=\{S: S \subseteq \mathcal{A}, S\ne \emptyset \}$.
Let $R$ be a subset of $\Omega$. Then it follows from \eqref{eq-thm1-wt(a)-2} that for $a\in \fqm^*$,
$wt(c_{a})=q^{m-1}-\sum_{S\in R}(-1)^{|S|-1}q^{|\cap S|-1}$ if $a\in \Psi_{R}$, where
\begin{eqnarray} \label{eq-thm1-wt-Psi-R}
\Psi_{R}:=\{a\in \fqm^*: a\in \Delta_{\cap S}^{\bot} \mbox{ for any } S\in \Omega \setminus R  \mbox{ and }  a\notin \Delta_{\cap S}^{\bot} \mbox{ for any } S\in R\}.
\end{eqnarray}
Next we compute the value of $|\Psi_{R}|$.
Note that for $S\in \Omega$, $\Delta_{\cap S}^{\bot}$ is a $|\cap S|$-dimensional $\fq$-subspace of $\fqm$ spanned by $\{\beta_{i}: i \in [m] \backslash \cap S\}$, i.e., $\Delta_{\cap S}^{\bot}=\mbox{\textbf{span}}\{\beta_{i}: i \in [m] \setminus \cap S\}$, and the intersection of some subspaces of $\fqm$ is also a subspace of $\fqm$. When $R= \emptyset$, $wt(c_{a})=q^{m-1}$ if $a\in \Psi_{R}$, and we have
\begin{align} \label{eq-thm1-wt-Fre-1}
|\Psi_{R}|=&|\{a\in \fqm: a\in \Delta_{\cap S}^{\bot} \mbox{ for any } S\in \Omega \}|-1
=|\cap_{S\in \Omega}\Delta_{\cap S}^{\bot}|-1\nonumber\\
=&|\Delta_{\cup_{S\in \Omega}(\cap S)}^{\bot}|-1=q^{m-|\cup_{S\in \Omega}(\cap S)|}-1
=q^{m-|\cup_{i=1}^{h}A_{i}|}-1
\end{align}
since $\cap_{S\in \Omega}\Delta_{\cap S}^{\bot}=\mbox{\textbf{span}}\{\beta_{i}: i \in [m] \setminus \cup_{ S\in \Omega}(\cap S)\}=\Delta_{\cup_{S\in \Omega }(\cap S)}^{\bot}$.
When $R\not = \emptyset$, by \eqref{eq-thm1-wt-Psi-R} and the fact that $\cap_{S\in \Omega \setminus R}\Delta_{\cap S}^{\bot}=\Delta_{\cup_{S\in \Omega \setminus R}(\cap S)}^{\bot}$, it gives
\begin{align}\label{eq-thm1-wt-Fre-2}
\Psi_{R}&=\cap_{S\in \Omega \setminus R}\Delta_{\cap S}^{\bot}\setminus
\cup_{L\in R}((\cap_{S\in \Omega \setminus R}\Delta_{\cap S}^{\bot})\cap\Delta_{\cap L}^{\bot})\nonumber\\
&=\Delta_{\cup_{S\in \Omega \setminus R}(\cap S)}^{\bot}\setminus
\cup_{L\in R}(\Delta_{\cap L}^{\bot}\cap\Delta_{\cup_{S\in \Omega \setminus R}(\cap S)}^{\bot})\nonumber\\
&=\Delta_{\cup_{S\in \Omega \setminus R}(\cap S)}^{\bot}\setminus
\cup_{L\in R}(\Delta_{(\cap L) \cup (\cup_{S\in \Omega \setminus R}(\cap S))}^{\bot}).
\end{align}

Notice that $\cup_{L\in R}(\Delta_{(\cap L) \cup (\cup_{S\in \Omega \setminus R}(\cap S))}^{\bot})$ is in fact the union of some subspaces $\Delta_{(\cap L) \cup (\cup_{S\in \Omega \setminus R}(\cap S))}^{\bot}$ of $\fqm$. Then, similar to the computation of $\Delta$ in \eqref{Delta-card}, it follows from  \eqref{eq-thm1-wt-Fre-2} that
\begin{eqnarray} \label{eq-thm1-Psi-R-value}
|\Psi_{R}|=q^{m-|\cup_{S\in \Omega \setminus R}(\cap S)|}-
\sum_{\emptyset \not=E\subseteq R}(-1)^{|E|-1}q^{m-|(\cup_{L\in E}(\cap L)) \cup (\cup_{S\in \Omega \setminus R}(\cap S))|},
\end{eqnarray}
since $\cap_{L \in E} \Delta_{(\cap L) \cup (\cup_{S\in \Omega \setminus R}(\cap S))}^{\bot}=
\Delta_{ (\cup_{L\in E}(\cap L)) \cup (\cup_{S\in \Omega \setminus R}(\cap S))}^{\bot}$ for any nonempty subset $E$ of $R$.

Then, as $R$ runs through all the subsets of $\Omega$, we get the weight enumerator of
$\C_{\overline{\Delta}^{c}}$ as follows:
\[\sum_{\emptyset \not=R\subseteq \Omega} |\Psi_{R}|
        z^{q^{m-1}-\sum_{S\in R}(-1)^{|S|-1}q^{|\cap S|-1}}+(q^{m-|\cup_{i=1}^{h}A_{i}|}-1)z^{q^{m-1}}+1.\]
This completes the proof.
\end{proof}

\begin{remark} \label{remark-thm1-d-optimal}
Here we show that the given condition $|\Delta|-1+(q-1)(v(T)+1)>qT-\ell(T)$ in 3) of Theorem \ref{optimalcode-anti-sc-pro} for $\C_{\overline{\Delta}^{c}}$ to be distance-optimal can be easily satisfied and consequently many distance-optimal linear codes over $\fq$ can be produced from $\C_{\overline{\Delta}^{c}}$. Note that
$qT-|\Delta|=\sum_{S\subseteq \mathcal{A}, |S|\geq 2}(-1)^{|S|-1}q^{|\cap S|}$ whose value heavily relies on those of
$|A_{i}\cap A_{j}|$ for $1\leq i < j \leq h$. By the definition, $v(T)\geq |A_{1}|$ and $\ell(T)\leq h$. Thus $|\Delta|-1+(q-1)(v(T)+1)>qT-\ell(T)$ can be easily satisfied if $|A_{1}|$ is large enough and $|A_{i}\cap A_{j}|$'s are small enough. This also can be verified by the discussion for the cases $h=1,2,3$ in Section \ref{sec4}.
\end{remark}

\begin{remark} \label{remark-thm1-wt}
The given formula in 4) of Theorem \ref{optimalcode-anti-sc-pro} to compute the weight distribution of $\C_{\overline{\Delta}^{c}}$  is completely computable for a given $\Delta$ with support $\mathcal{A}=\{A_{1},A_{2},\ldots,A_{h}\}$ although the expression seems not so simple. This can be verified in Section \ref{sec4} in which we completely determine the weight distributions of $\C_{\overline{\Delta}^{c}}$ for the three cases of the simplicial complexes with one, two and three maximal elements according to this formula. Thus we say that the weight distribution of $\C_{\overline{\Delta}^{c}}$ is completely determined in Theorem \ref{optimalcode-anti-sc-pro}. Moreover, the weight distribution of $\C_{\overline{\Delta}^{c}}$ can be computed efficiently by the formula given in 4) of Theorem \ref{optimalcode-anti-sc-pro} with the following property: for $S_{1}\in \Omega$ and $S_{2}\in R \subseteq \Omega$, one has $|\Psi_{R}|= 0$ if $S_{1}\notin R$ for some $S_{1}\subseteq S_{2}$, since $a\in \Delta_{\cap S_{1}}^{\bot}$ implies $a\in \Delta_{\cap S_{2}}^{\bot}$ if $S_{1} \subseteq S_{2}$.

\end{remark}


In the following corollary, we take a more in-depth discussion on the case that $|A_{i}\cap A_{j}| =0$ for all $1\leq i< j \leq h$  for the code $\C_{\overline{\Delta}^{c}}$ in Theorem \ref{optimalcode-anti-sc-pro}. As a result, a more simple expression for the weight distribution of $\C_{\overline{\Delta}^{c}}$ is given and a necessary and sufficient condition for $\C_{\overline{\Delta}^{c}}$ to be a near Griesmer code is derived in this case.
\begin{cor} \label{optimalcode-anti-sc-pro-case1}
Let $\Delta$ be a simplicial complex of $\fqm$ with the support $\mathcal{A}=\{A_{1},A_{2},\ldots,A_{h} \}$, where $1 \leq |A_{1}| \leq |A_{2}| \leq \cdots \leq |A_{h}|<m$. Assume that $|A_{i}\cap A_{j}| =0$ for $1\leq i< j \leq h$. Denote $T=\sum_{1\leq i\leq h} q^{|A_{i}|-1}$. Let $\C_{\overline{\Delta}^{c}}$ be defined as in \eqref{CD}. Then $\C_{\overline{\Delta}^{c}}$ is an at most $2^h$-weight $[(q^m-\sum_{i=1}^{h}q^{|A_{i}|}+h-1)/(q-1),m,q^{m-1}-T]$ linear code with weight enumerator
\[\sum_{\emptyset \ne R\subseteq [h]}(q^{m-\sum_{i\in [h]\backslash R}|A_{i}|}-\sum_{\emptyset \not=E\subseteq R}(-1)^{|E|-1}q^{m-\sum_{i\in E}|A_{i}|-\sum_{i\in [h]\backslash R}|A_{i}|})z^{q^{m-1}-\sum_{i\in R}q^{|A_{i}|-1}}+(q^{m-\sum_{i=1}^{h}|A_{i}|}-1)z^{q^{m-1}}+1.\]
Moreover, we have the followings:
\begin{itemize}
  \item [1)]$\C_{\overline{\Delta}^{c}}$ is a Griesmer code if and only if at most $q-1$ of $|A_{i}|$'s are the same;
  \item [2)]$\C_{\overline{\Delta}^{c}}$ is a near Griesmer code if and only if  $\ell(T)=h-(q-1)$; and
  \item [3)]$\C_{\overline{\Delta}^{c}}$ is distance-optimal if $\ell(T)+ (q-1)(v(T)+1)>h$. Specially, when $|A_{i}|=\epsilon$ for $1\leq i \leq h$, where $\epsilon$ is a positive integer, it is distance-optimal if $\ell(h)+(q-1)(v(h)+\epsilon)>h$.
\end{itemize}
\end{cor}

\begin{proof}
Clearly, the conditions $A_{i}\backslash (\cup_{1\leq j \leq h, j\ne i}A_{j}) \ne \emptyset$ and
$q^{m} > \sum_{1\leq i\leq h} q^{|A_{i}|}$ in Theorem \ref{optimalcode-anti-sc-pro} are always satisfied under the assumption that   $|A_{i}\cap A_{j}| =0$ for $1\leq i< j \leq h$.
Recall that $|\Delta|=|\cup_{i=1}^{h}\Delta_{A_{i}}|$ where $\Delta_{A_{i}}=\langle \{A_{i}\}\rangle$.
Since $\Delta_{A_{i}} \cap \Delta_{A_{j}} = \{0\}$ due to $|A_{i}\cap A_{j}| =0$, we have $|\Delta|=\sum_{i=1}^{h}q^{|A_{i}|}-h+1$.
Then the parameters of $\C_{\overline{\Delta}^{c}}$ follow directly from 1) of Theorem \ref{optimalcode-anti-sc-pro}.
Moreover, by 2) and 3) of Theorem \ref{optimalcode-anti-sc-pro}, $\C_{\overline{\Delta}^{c}}$ is a Griesmer code if and only if at most $q-1$ of $|A_{i}|$'s are the same and it is distance-optimal if $\ell(T)+ (q-1)(v(T)+1)>h$. By Lemma \ref{g(m,d)}, one can check that $n-g(m,d)=1$ if and only if $\ell(T)=h-(q-1)$, where $n=(q^m-\sum_{i=1}^{h}q^{|A_{i}|}+h-1)/(q-1)$ and $d=q^{m-1}-T$.
Specially, when $|A_{i}|=\epsilon$ for $1\leq i \leq h$, which implies $T=hq^{\epsilon-1}$, it can be verified that $\ell(T)=\ell(h)$ and $v(T)=v(h)+\epsilon-1$. Thus, when $|A_{i}|=\epsilon$ for $1\leq i \leq h$, $\C_{\overline{\Delta}^{c}}$ is distance-optimal if $\ell(h)+ (q-1)(v(h)+\epsilon)>h$. This proves 1), 2) and 3). Next we compute the weight enumerator of $\C_{\overline{\Delta}^{c}}$ according to the proof of 4) of Theorem \ref{optimalcode-anti-sc-pro}.

Recall from the proof of 4) of Theorem \ref{optimalcode-anti-sc-pro} that for $a\in \fqm^*$,
$wt(c_{a})=q^{m-1}-\sum_{S\in R}(-1)^{|S|-1}q^{|\cap S|-1}$ if $a\in \Psi_{R}$,
where $\Omega=\{S: S \subseteq \mathcal{A}, S\ne \emptyset \}$, $R$ is a subset of $\Omega$ and $\Psi_{R}$ is defined as in \eqref{eq-thm1-wt-Psi-R}. Notice that for $S\in \Omega$ with $|S|\geq 2$, one has $\Delta_{\cap S}^{\bot}=\fqm$ since $\cap S=\emptyset$ and $\Delta_{\cap S}=\{0\}$. If there is some $S$ with $|S|\geq 2$ in $R$, then $|\Psi_{R}|=0$.
Thus we can rewrite simply that for $a\in \fqm^*$,
$wt(c_{a})=q^{m-1}-\sum_{i\in R}q^{|A_{i}|-1}$ if $a\in \Psi_{R}$, where $R$ is a subset of $[h]$, $\Psi_{R}=\{a\in \fqm^*: a\in \Delta_{A_{i}}^{\bot} \mbox{ for any } i\in [h] \setminus R  \mbox{ and }  a\notin \Delta_{A_{i}}^{\bot} \mbox{ for any } i\in R\}$ and $\Delta_{A_{i}}=\langle \{A_{i}\}\rangle$. Specially, when $R=\emptyset$, similar to the computation in the proof of 4) of Theorem \ref{optimalcode-anti-sc-pro}, we also have $wt(c_{a})=q^{m-1}$ if $a\in \Psi_{R}$ and $|\Psi_{R}|=q^{m-|\cup_{i=1}^{h}A_{i}|}-1$. When $R\ne \emptyset$, similar to the computation of $|\Psi_{R}|$ in the proof of 4) of Theorem \ref{optimalcode-anti-sc-pro}, we have
\[|\Psi_{R}|=q^{m-\sum_{i\in [h] \setminus R}|A_{i}|}-\sum_{\emptyset \not=E\subseteq R}(-1)^{|E|-1}q^{m-\sum_{i\in E}|A_{i}|-\sum_{i\in [h] \setminus R}|A_{i}|}.\]
Therefore, as $R$ runs through all subsets of $[h]$, the weight enumerator of $\C_{\overline{\Delta}^{c}}$ can be expressed as
\[\sum_{\emptyset \ne R \subseteq [h]}|\Psi_{R}|z^{q^{m-1}-\sum_{i\in R}q^{|A_{i}|-1}}+(q^{m-\sum_{i=1}^{h}|A_{i}|}-1)z^{q^{m-1}}+1.\]
Due to the $2^h$ choices of the subset $R$ of $[h]$, $\C_{\overline{\Delta}^{c}}$ is at most $2^h$-weight.
This completes the proof.
\end{proof}

\begin{remark} \label{remark-SScodes}
The Griesmer codes in Corollary \ref{optimalcode-anti-sc-pro-case1} (or Theorem \ref{optimalcode-anti-sc-pro}) are indeed the Solomon-Stiffler codes, since $\overline{\Delta}^{c}$ can be viewed as the complement of the union of $h$ disjoint projective subspaces of dimensions $|A_{i}|-1$ of $\PG(m-1, q)$ when $|A_{i}\cap A_{j}|=0$ for any $1\leq i<j \leq h$. Definitely, for the other cases (not the Griesmer codes), our codes $\C_{\overline{\Delta}^{c}}$ in Corollary \ref{optimalcode-anti-sc-pro-case1} and Theorem \ref{optimalcode-anti-sc-pro} are different from the Solomon-Stiffler codes.
\end{remark}

\begin{remark}
Notice that the necessary and sufficient condition $\ell(T)=h-(q-1)$ for $\C_{\overline{\Delta}^{c}}$ to be a near Griesmer code in Corollary \ref{optimalcode-anti-sc-pro-case1} can be easily satisfied by selecting proper $A_{i}$'s.
Moreover, the sufficient condition $\ell(T)+(q-1)(v(T)+1)>h$ for $\C_{\overline{\Delta}^{c}}$ to be distance-optimal can be easily satisfied if $|A_{1}|$ is large enough since $1\leq \ell(T)\leq h$ and $v(T)\geq |A_{1}|$, and consequently many distance-optimal linear codes can be produced in Corollary \ref{optimalcode-anti-sc-pro-case1} besides (near) Griesmer codes.
\end{remark}


Here we will investigate the parameters and weight distribution of the projective linear code $\C_{\overline{\Delta}^*}$ defined by \eqref{CD}, where $\Delta^*=\Delta\setminus\{0\}$ and $\overline{\Delta}^*$ is defined by $\Delta^{*}=\fq^{*}\overline{\Delta}^{*}=\{yz: y \in \fq^{*} \mbox{ and } z \in \overline{\Delta}^{*}\}$ in which $z_{i}/z_{j} \notin \fq^*$ for distinct $z_{i}$ and $z_{j}$ in $\overline{\Delta}^{*}$.

\begin{prop} \label{prop-1}
Let $\Delta$ be a simplicial complex of $\fqm$ with the support $\mathcal{A}=\{A_{1},A_{2},\ldots,A_{h}\}$, where $1 \leq |A_{1}| \leq |A_{2}| \leq \cdots \leq |A_{h}|<m$. Assume that $A_{i}\backslash (\cup_{1\leq j \leq h, j\ne i}A_{j}) \ne \emptyset$ for any $1 \leq i \leq h$ and $q^{m} > \sum_{1\leq i\leq h} q^{|A_{i}|}$. Let the multiset $\{*\,\,wt(c_{a}): a\in \fqm^* \,\,*\} \cup \{*\,\, 0 \,\,*\}$ be the weight distribution of $\C_{\overline{\Delta}^{c}}$ given as in Theorem \ref{optimalcode-anti-sc-pro} and define the multiset
\begin{equation} \label{eq-remark-1}
\{*\,\,q^{m-1}-wt(c_{a}): a\in \fqm^* \,\,*\} \cup \{*\,\, 0 \,\,*\}.
\end{equation}
Let $\C_{\overline{\Delta}^*}$ be defined as in \eqref{CD}. Then $\C_{\overline{\Delta}^*}$ has parameters $[(|\Delta|-1)/(q-1),|\cup_{i=1}^{h}A_{i}|,q^{|A_{1}|-1}]$, where $|\Delta|=\sum_{\emptyset \not=S\subseteq \mathcal{A}}(-1)^{|S|-1}q^{|\cap S|}$, and its weight distribution is given by
\[\mbox{$q^{m-1}-wt(c_{a})$ with multiplicity $e_{a}/q^{m-|\cup_{i=1}^{h}A_{i}|}$}\]
for all $q^{m-1}-wt(c_{a})$ in \eqref{eq-remark-1}, where $e_{a}$ is the multiplicity of $q^{m-1}-wt(c_{a})$ in the multiset of \eqref{eq-remark-1}.
\end{prop}
\begin{proof}
It's obvious that $\C_{\overline{\Delta}^*}$ has length $n'=(|\Delta|-1)/(q-1)$ where $|\Delta|$ is given as in \eqref{Delta-card}. Next we study the Hamming weight of the codewords $c'_{a}$ in $\C_{\overline{\Delta}^*}$ according to the proof of Theorem \ref{optimalcode-anti-sc-pro} and the relation between $\C_{\overline{\Delta}^*}$ and $\C_{\overline{\Delta}^{c}}$.
It can be readily verified that  $wt(c'_{a})+wt(c_{a})=q^{m-1}$ for $a\in \fqm^*$, where $wt(c'_{a})$ (resp. $wt(c_{a})$) denotes the Hamming weight of the codeword $c'_{a}$ (resp. $c_{a}$) in $\C_{\overline{\Delta}^*}$ (resp. $\C_{\overline{\Delta}^{c}}$). Clearly, the weight distribution of $\C_{\overline{\Delta}^{c}}$ is given by the multiset $\{*\,\,wt(c_{a}): a\in \fqm^* \,\,*\} \cup \{*\,\, 0 \,\,*\}$ which can be obtained by 4) of Theorem \ref{optimalcode-anti-sc-pro} and then the multiset $\{*\,\,wt(c'_{a}): a\in \fqm^* \,\,*\} \cup \{*\,\, 0 \,\,*\}$ on the code $\C_{\overline{\Delta}^*}$ is given by \eqref{eq-remark-1}.

Now we need to compute the multiplicity of $0$ in \eqref{eq-remark-1} and the minimum nonzero value in the set \eqref{eq-remark-1} to determine the dimension and minimum distance of $\C_{\overline{\Delta}^*}$.

Due to \eqref{eq-thm1-wt(a)} in the proof of 1) of Theorem \ref{optimalcode-anti-sc-pro}, one has $wt(c'_{a})=(|\Delta|-\Theta)/(q-1)$, where $|\Delta|-\Theta$ is defined as in the proof of Theorem \ref{optimalcode-anti-sc-pro}. Moreover, according to the proof of 4) of Theorem \ref{optimalcode-anti-sc-pro}, we also have
$wt(c'_{a})=\sum_{\emptyset \not=S\in R}(-1)^{|S|-1}q^{|\cap S|-1}$ if $a\in \Psi_{R}$, where $R$ is a subset of $\Omega$, $\Omega=\{S: S \subseteq \mathcal{A}, S\ne \emptyset \}$ and $\Psi_{R}$ is defined as in \eqref{eq-thm1-wt-Psi-R}. Observe that $wt(c'_{a})=0$ if $a\in \Psi_{R}$ with $R= \emptyset$. Note that for $S_{1},\,\,S_{2}\in \Omega$, one can verify that $a\in \Delta_{\cap S_{1}}^{\bot}$ implies $a\in \Delta_{\cap S_{2}}^{\bot}$ if $S_{1} \subseteq S_{2}$. Thus, when $R \ne \emptyset$, we conclude that there must exist some $A_{i}$ for $1\leq i \leq h$ such that $\{A_{i}\} \in R$ if $|\Psi_{R}|\ne 0$. Since it's meaningless for $|\Psi_{R}|=0$, we always assume that $|\Psi_{R}|\ne 0$ for the subsequent proof. We first consider the case that $a\in \Psi_{R}$ and $\{A_{1}\} \in R$, which implies that $a\notin \Delta_{A_{1}}^{\bot}$ where $\Delta_{A_{1}}=\langle \{A_{1}\}\rangle$. It follows from  \eqref{eq-thm1-DeltaOmega-2} that
\begin{equation} \label{eq-prop1-2}
wt(c'_{a})=(|\Delta|-\Theta)/(q-1)
=\frac{1}{q-1}(f_{a}(\{A_{1}\})+\sum_{i=2}^{h}(f_{a}(\{A_{i}\})-\sum_{\emptyset \not=S\subseteq \mathcal{A}_{i}}(-1)^{|S|-1}f_{a}(S))).
\end{equation}
Here $f_{a}(\{A_{1}\})=(q-1)q^{|A_{1}|-1}$ due to $a\notin \Delta_{A_{1}}^{\bot}$.
Moreover, similar to the computation of \eqref{eq-thm1-wt(a)} and \eqref{eq-thm1-DeltaOmega-1}, one can check that $f_{a}(\{A_{i}\})-\sum_{\emptyset \not=S\subseteq \mathcal{A}_{i}}(-1)^{|S|-1}f_{a}(S)$ is the Hamming weight of the codeword $\widehat{c}_{a}$ in the code $\C_{\Delta_{A_{i}}\setminus \Delta_{\mathcal{A}_{i}}}$ defined by $\C_{\Delta_{A_{i}}\setminus \Delta_{\mathcal{A}_{i}}}=\{\widehat{c}_{a}=(\Tr_{q}^{q^{m}}(ax))_{x\in \Delta_{A_{i}}\setminus \Delta_{\mathcal{A}_{i}}}: a\in \fqm\}$. This implies that $f_{a}(\{A_{i}\})-\sum_{\emptyset \not=S\subseteq \mathcal{A}_{i}}(-1)^{|S|-1}f_{a}(S) \geq 0$ for $2 \leq i \leq h$. Thus we have $wt(c'_{a})\geq q^{|A_{1}|-1}$ if $a\in \Psi_{R}$ and $\{A_{1}\} \in R$. Similarly, we can prove that $wt(c'_{a})\geq q^{|A_{i}|-1}$ if $a\in \Psi_{R}$ and $\{A_{i}\} \in R$ for $2\leq i \leq h$ since the condition $1 \leq |A_{1}| \leq |A_{2}| \leq \cdots \leq |A_{h}|<m$ has no effect on the above computation. Then one has $wt(c'_{a})\geq q^{|A_{1}|-1}>0$ for any $a\in \Psi_{R}$ with $R\ne \emptyset$ and $|\Psi_{R}|\ne 0$ due to $1 \leq |A_{1}| \leq |A_{2}| \leq \cdots \leq |A_{h}|<m$.
This accordingly proves that $wt(c'_{a})=0$ if and only if $a\in \Psi_{R}$, $R= \emptyset$ or $a=0$, which implies the multiplicity of $0$ in the multiset of \eqref{eq-remark-1} is $q^{m-|\cup_{i=1}^{h}A_{i}|}$ due to \eqref{eq-thm1-wt-Fre-1} and consequently the dimension of $\C_{\overline{\Delta}^*}$ is $|\cup_{i=1}^{h}A_{i}|$. Then the weight distribution of $\C_{\overline{\Delta}^*}$ follows from \eqref{eq-remark-1}.
In addition, it can be readily verified that $wt(c'_{a})= q^{|A_{1}|-1}$ if $a\in \Psi_{R}$ and $R=\{\{A_{1}\}\}$, and  in this case $|\Psi_{R}|=q^{m-|\cup_{i=2}^{h}A_{i}|}-q^{m-|\cup_{i=1}^{h}A_{i}|}>0$ due to \eqref{eq-thm1-Psi-R-value}. This indicates that the minimum distance of $\C_{\overline{\Delta}^*}$ is $q^{|A_{1}|-1}$.
 This completes the proof.
\end{proof}

\begin{remark} \label{remark-prop-1}
$\C_{\overline{\Delta}^*}$ in Proposition \ref{prop-1} is a trivial code with parameters $[h,h,1]$ (the full vector space $\fq^h$) if $|A_{h}|=1$ and it is the well-known simplex code with parameters $[(q^{|A_{1}|}-1)/(q-1), |A_{1}|, q^{|A_{1}|-1}]$ if $h=1$. With some detailed computation, when $|A_{h}|>1$ and $h>1$, one can verify that $\C_{\overline{\Delta}^*}$ cannot be a Griesmer code, i.e., $n\ne g(k,d)$, and $n<g(k,d+1)$ only if $(q, h)=(2,2)$ and $|A_{1}|=|A_{2}|=2$, which implies that
$\C_{\overline{\Delta}^*}$ is distance-optimal if $(q, h)=(2,2)$ and $|A_{1}|=|A_{2}|=2$.
\end{remark}


\section{Explicit constructions of optimal linear codes for $h=1,2,3$}  \label{sec4}
In this section, we take a more in-depth discussion on the cases $h=1,2,3$ of Theorem \ref{optimalcode-anti-sc-pro} in order to show that many infinite families of optimal linear codes can be obtained from Theorem \ref{optimalcode-anti-sc-pro}, including Griesmer codes, near Griesmer codes and distance-optimal codes.
Moreover, for these three cases, we will present the parameters and weight distributions of $\C_{\overline{\Delta}^{c}}$ more explicitly and characterize the optimality of $\C_{\overline{\Delta}^{c}}$ in detail.
Notice that the cases $h=1,2,3$ correspond exactly to the cases of simplicial complexes $\Delta$ with one, two and three maximal elements, respectively.

\begin{thm} \label{optimalcode-anti-sc-h=1}
Let $\Delta$ be a simplicial complex of $\fqm$ with exactly one maximal element and its support is $\{A\}$ with $A\subseteq [m]$ and $1 \leq |A|<m$.
Then $\C_{\overline{\Delta}^{c}}$ defined by \eqref{CD} is a $2$-weight $[(q^m-q^{|A|})/(q-1),m,q^{m-1}-q^{|A|-1}]$ linear code with weight distribution
\begin{center}
\newcommand{\tabincell}[2]{\begin{tabular}{@{}#1@{}}#2\end{tabular}}
\begin{tabular}{|l|l|}
\hline Weight $w$ & Multiplicity $A_{w}$\\ \hline\hline
                                           $0$ & $1$ \\ \hline
$q^{m-1}$               & $q^{m-|A|}-1$ \\ \hline
$q^{m-1}-q^{|A|-1}$ & $q^{m}-q^{m-|A|}$ \\ \hline
\end{tabular}
\end{center}
and it is a Griesmer code.
\end{thm}

\begin{proof}
By Corollary \ref{optimalcode-anti-sc-pro-case1}, the parameters of $\C_{\overline{\Delta}^{c}}$ follow directly and it is a Griesmer code. Moreover, its weight enumerator is given by
$1+(q^{m-|A|}-1)z^{q^{m-1}}+(q^{m}-q^{m-|A|})z^{q^{m-1}-q^{|A|-1}}$ due to Corollary \ref{optimalcode-anti-sc-pro-case1}.
This completes the proof.
\end{proof}

\begin{example}
Let $q=3$, $m=5$, $A=\{1,2\}$. Magma experiments show that $\C_{\overline{\Delta}^{c}}$ is a $[117,5,78]$ linear code over ${\mathbb F}_{3}$ with weight enumerator $1+216z^{78}+26z^{81}$, which is consistent with our result in Theorem \ref{optimalcode-anti-sc-h=1}. This code is a Griesmer code by the Griesmer bound and is optimal due to \cite{GMB}.
\end{example}

\begin{thm} \label{optimalcode-anti-sc-h=2}
Let $\Delta$ be a simplicial complex of $\fqm$ with the support $\mathcal{A}=\{A_{1},A_{2}\}$, where $1 \leq |A_{1}| \leq |A_{2}|<m$.
Assume that $q^{m} > q^{|A_{1}|}+q^{|A_{2}|}$. Let $T=q^{|A_{1}|-1}+q^{|A_{2}|-1}$. Then $\C_{\overline{\Delta}^{c}}$ defined by \eqref{CD} is an at most $5$-weight $[(q^m-q^{|A_{1}|}-q^{|A_{2}|}+q^{|A_{1}\cap A_{2}|})/(q-1),m,q^{m-1}-q^{|A_{1}|-1}-q^{|A_{2}|-1}]$ linear code and its weight distribution is given by
\begin{center}
\newcommand{\tabincell}[2]{\begin{tabular}{@{}#1@{}}#2\end{tabular}}
\begin{tabular}{|l|l|}
\hline Weight $w$ & Multiplicity $A_{w}$\\ \hline\hline
                                           $0$ & $1$ \\ \hline
$q^{m-1}$               & $q^{m-|A_{1}\cup A_{2}|}-1$ \\ \hline
$q^{m-1}-q^{|A_{2}|-1}$ & $q^{m-|A_{1}|}-q^{m-|A_{1}\cup A_{2}|}$ \\ \hline
$q^{m-1}-q^{|A_{1}|-1}$ & $q^{m-|A_{2}|}-q^{m-|A_{1}\cup A_{2}|}$ \\ \hline
$q^{m-1}-q^{|A_{1}|-1}-q^{|A_{2}|-1}$ & $q^{m-|A_{1}\cap A_{2}|}-q^{m-|A_{1}|}-q^{m-|A_{2}|}+q^{m-|A_{1}\cup A_{2}|}$ \\\hline
$q^{m-1}-q^{|A_{1}|-1}-q^{|A_{2}|-1}+q^{|A_{1}\cap A_{2}|-1}$ & $q^{m}-q^{m-|A_{1}\cap A_{2}|}$\\ \hline
\end{tabular}
\end{center}
Moreover, we have the followings:
\begin{enumerate}
\item [1)] When $|A_{1}\cap A_{2}|=0$ and $|A_{1}|=|A_{2}|$, $\C_{\overline{\Delta}^{c}}$ is a near Griesmer code (also distance-optimal) if $q=2$ and it is a Griesmer code if $q>2$. It reduces to a $3$-weight code in this case.
\item [2)] When $|A_{1}\cap A_{2}|=0$ and $|A_{1}|< |A_{2}|$, $\C_{\overline{\Delta}^{c}}$ is a Griesmer code and it reduces to a $4$-weight code.
\item [3)] When $|A_{1}\cap A_{2}|>0$ and $|A_{1}|=|A_{2}|$, $\C_{\overline{\Delta}^{c}}$ is distance-optimal if $\ell(T)+(q-1)(v(T)+1)>q^{|A_{1}\cap A_{2}|}+1$ and it reduces to a $4$-weight code. Specially, $\C_{\overline{\Delta}^{c}}$ is a near Griesmer code if $q>2$ and $|A_{1}\cap A_{2}|=1$.
\item [4)] When $|A_{1}\cap A_{2}|>0$ and $|A_{1}|< |A_{2}|$, $\C_{\overline{\Delta}^{c}}$ is distance-optimal if $(q-1)|A_{1}|+1>q^{|A_{1}\cap A_{2}|}$. Specially, $\C_{\overline{\Delta}^{c}}$ is a near Griesmer code if $|A_{1}\cap A_{2}|=1$.
\end{enumerate}
\end{thm}

\begin{proof}
By the definition of simplicial complexes, the condition $A_{i}\backslash (\cup_{1\leq j \leq h, j\ne i}A_{j}) \ne \emptyset$  in Theorem \ref{optimalcode-anti-sc-pro} always holds for $h=2$.
Then the parameters of $\C_{\overline{\Delta}^{c}}$ follow directly from 1) of Theorem \ref{optimalcode-anti-sc-pro}. Due to 4) of Theorem \ref{optimalcode-anti-sc-pro}, we can determine each possible Hamming weight of the codewords in $\C_{\overline{\Delta}^{c}}$ and its corresponding multiplicity as $R$ runs through all subsets of $\Omega=\{S: S \subseteq \mathcal{A}, S\ne \emptyset \}=\{\,\{A_{1}\}, \{A_{2}\}, \{A_{1},A_{2}\}\,\}$. Then the weight distribution can be obtained by some careful computation. Next we discuss the optimality of $\C_{\overline{\Delta}^{c}}$ case by case.

Note that the length of $\C_{\overline{\Delta}^{c}}$ is $n=\frac{1}{q-1}(q^m-q^{|A_{1}|}-q^{|A_{2}|}+q^{|A_{1}\cap A_{2}|})$ and the minimum distance of $\C_{\overline{\Delta}^{c}}$ is $d=q^{m-1}-q^{|A_{1}|-1}-q^{|A_{2}|-1}$.
Then by Lemma \ref{g(m,d)} we have
\begin{equation} \label{eq-cor3-1}
g(m,d)= \frac{1}{q-1}(q^m-q^{|A_{1}|}-q^{|A_{2}|}+\ell(T)-1)
\end{equation}
and
\begin{equation} \label{eq-cor3-2}
g(m,d+1)= \frac{1}{q-1}(q^m-q^{|A_{1}|}-q^{|A_{2}|}+\ell(T)-1)+v(T)+1,
\end{equation}
where $T=q^{|A_{1}|-1}+q^{|A_{2}|-1}$. As for $\ell(T)$ and $v(T)$, one can easily verify that
\begin{equation}  \label{eq-cor3-3}
\ell(T)=\left\{\begin{array}{ll}
1,     &   \mbox{ if $q=2$ and $|A_{1}| = |A_{2}|$};\\
2,  &   \mbox{ otherwise},
\end{array} \right.
\end{equation}
and
\begin{equation}  \label{eq-cor3-4}
v(T)=\left\{\begin{array}{ll}
|A_{1}|,     &   \mbox{ if $q=2$ and $|A_{1}| = |A_{2}|$};\\
|A_{1}|-1,  &   \mbox{ otherwise}.
\end{array} \right.
\end{equation}

Therefore we have the followings:

1). In the case $|A_{1}\cap A_{2}|=0$ and $|A_{1}|=|A_{2}|$, we have $n-g(m,d)=\frac{1}{q-1}(2-\ell(T))$. This together with \eqref{eq-cor3-3} implies that $\C_{\overline{\Delta}^{c}}$ is a near Griesmer code if $q=2$ and it is a Griesmer code if $q>2$.
Moreover, $\C_{\overline{\Delta}^{c}}$ is always distance-optimal since $g(m,d+1)-n=\frac{1}{q-1}(\ell(T)-2)+v(T)+1>0$.
Notice that it reduces to a $3$-weight code.

2). In the case $|A_{1}\cap A_{2}|=0$ and $|A_{1}|< |A_{2}|$, $\C_{\overline{\Delta}^{c}}$ is a  Griesmer code since $n-g(m,d)=\frac{1}{q-1}(2-\ell(T))=0$. Observe that it reduces to a $4$-weight code.

3). In the case $|A_{1}\cap A_{2}|>0$ and $|A_{1}|=|A_{2}|$, we have $n-g(m,d)=\frac{1}{q-1}(q^{|A_{1}\cap A_{2}|}+1-\ell(T))>0$. Specially, if $q>2$ and $|A_{1}\cap A_{2}|=1$, it gives $n-g(m,d)=1 $, which implies that $\C_{\overline{\Delta}^{c}}$ is a near Griesmer code. By \eqref{eq-cor3-2}, $\C_{\overline{\Delta}^{c}}$ is distance-optimal if $g(m,d+1)-n=\frac{1}{q-1}(\ell(T)-1-q^{|A_{1}\cap A_{2}|})+v(T)+1>0$ which is equivalent to $\ell(T)+(q-1)(v(T)+1)>q^{|A_{1}\cap A_{2}|}+1$.  Notice that it reduces to a $4$-weight code.

4). In the case  $|A_{1}\cap A_{2}|>0$ and $|A_{1}|< |A_{2}|$, we have $n-g(m,d)=\frac{1}{q-1}(q^{|A_{1}\cap A_{2}|}-1)>0$. Clearly, $\C_{\overline{\Delta}^{c}}$ is a near Griesmer code if $|A_{1}\cap A_{2}|=1$. By \eqref{eq-cor3-2}, $\C_{\overline{\Delta}^{c}}$ is distance-optimal if $g(m,d+1)-n>0$, i.e., $(q-1)|A_{1}|+1>q^{|A_{1}\cap A_{2}|}$.
This completes the proof.
\end{proof}
%

\begin{remark}
The conditions $\ell(T)+(q-1)(v(T)+1)>q^{|A_{1}\cap A_{2}|}+1$ and $(q-1)|A_{1}|+1>q^{|A_{1}\cap A_{2}|}$ in 3) and 4) of Theorem \ref{optimalcode-anti-sc-h=2} for $\C_{\overline{\Delta}^{c}}$ to be distance-optimal can be easily satisfied if $|A_{1}|$ is large enough and $|A_{1}\cap A_{2}|$ is small enough due to \eqref{eq-cor3-3} and \eqref{eq-cor3-4} in the proof of Theorem \ref{optimalcode-anti-sc-h=2}. Thus many distance-optimal linear codes can be obtained since such $A_{1}$ and $A_{2}$ are abundant.
\end{remark}

\begin{example}   
Let $q=3$, $m=5$, $A_{1}=\{1,2\}$, $A_{2}=\{3,4\}$. Magma experiments show that $\C_{\overline{\Delta}^{c}}$ is a $[113,5,75]$ linear code over ${\mathbb F}_{3}$ with weight enumerator $1+192z^{75}+48z^{78}+2z^{81}$, which is consistent with our result in Theorem \ref{optimalcode-anti-sc-h=2}. This code is a Griesmer code by the Griesmer bound and is optimal due to \cite{GMB}.
\end{example}

\begin{example}  \label{exam-5} 
Let $q=3$, $m=5$, $A_{1}=\{1,2\}$, $A_{2}=\{2,3,4\}$. Magma experiments show that $\C_{\overline{\Delta}^{c}}$ is a $[105,5,69]$ linear code over ${\mathbb F}_{3}$ with weight enumerator $1+48z^{69}+162z^{70}+24z^{72}+6z^{78}+2z^{81}$, which is consistent with our result in Theorem \ref{optimalcode-anti-sc-h=2}. This code is optimal due to \cite{GMB}.
\end{example}

\begin{table}[!htb]\footnotesize
\caption{Weight distribution of $\C_{\overline{\Delta}^{c}}$ in Theorem \ref{optimalcode-anti-sc-h=3}}  \label{wd-h=3}
\centering
\newcommand{\tabincell}[2]{\begin{tabular}{@{}#1@{}}#2\end{tabular}}
\begin{tabular}{|l|l|}
\hline Weight $w$ & Multiplicity $A_{w}$              \\ \hline \hline
$0$ & $1$                                                  \\\hline
$q^{m-1}$               & $q^{m-|A_{1} \cup A_{2} \cup A_{3}|}-1$  \\\hline
$q^{m-1}-q^{|A_{1}|-1}$ & $q^{m-|A_{2} \cup A_{3}|}-q^{m-|A_{1} \cup A_{2} \cup A_{3}|}$   \\\hline
$q^{m-1}-q^{|A_{2}|-1}$ & $q^{m-|A_{1} \cup A_{3}|}-q^{m-|A_{1} \cup A_{2} \cup A_{3}|}$   \\\hline
$q^{m-1}-q^{|A_{3}|-1}$ & $q^{m-|A_{1} \cup A_{2}|}-q^{m-|A_{1} \cup A_{2} \cup A_{3}|}$   \\\hline
$q^{m-1}-q^{|A_{1}|-1}-q^{|A_{2}|-1}$ & $q^{m-|A_{3} \cup (A_{1} \cap A_{2})|}-q^{m-|A_{1} \cup A_{3}|}-q^{m-|A_{2} \cup A_{3}|}+q^{m-|A_{1} \cup A_{2} \cup A_{3}|}$  \\\hline
$q^{m-1}-q^{|A_{1}|-1}-q^{|A_{2}|-1}+q^{|A_{1} \cap A_{2}|-1}$ & $q^{m-|A_{3}|}-q^{m-|A_{3} \cup (A_{1} \cap A_{2})|}$ \\\hline
$q^{m-1}-q^{|A_{1}|-1}-q^{|A_{3}|-1}$ & $q^{m-|A_{2} \cup (A_{1} \cap A_{3})|}-q^{m-|A_{1} \cup A_{2}|}-q^{m-|A_{2} \cup A_{3}|}+q^{m-|A_{1} \cup A_{2} \cup A_{3}|}$ \\\hline
$q^{m-1}-q^{|A_{1}|-1}-q^{|A_{3}|-1}+q^{|A_{1} \cap A_{3}|-1}$ & $q^{m-|A_{2}|}-q^{m-|A_{2} \cup (A_{1} \cap A_{3})|}$ \\\hline
$q^{m-1}-q^{|A_{2}|-1}-q^{|A_{3}|-1}$ & $q^{m-|A_{1} \cup (A_{2} \cap A_{3})|}-q^{m-|A_{1} \cup A_{2}|}-q^{m-|A_{1} \cup A_{3}|}+q^{m-|A_{1} \cup A_{2} \cup A_{3}|}$ \\\hline
$q^{m-1}-q^{|A_{2}|-1}-q^{|A_{3}|-1}+q^{|A_{2} \cap A_{3}|-1}$ & $q^{m-|A_{1}|}-q^{m-|A_{1} \cup (A_{2} \cap A_{3})|}$ \\\hline
$q^{m-1}-\sum_{i=1}^{3}q^{|A_{i}|-1}$ &
\tabincell{l}{$q^{m-|(A_{1} \cap A_{2}) \cup (A_{1} \cap A_{3})\cup (A_{2} \cap A_{3})|}-q^{m-|A_{1}\cup(A_{2} \cap A_{3})|}-q^{m-|A_{2}\cup(A_{1} \cap A_{3})|}-q^{m-|A_{3}\cup(A_{1} \cap A_{2})|}$\\$+q^{m-|A_{1}\cup A_{2}|}+q^{m-|A_{1}\cup A_{3}|}+q^{m-|A_{2}\cup A_{3}|}-q^{m-|A_{1}\cup A_{2}\cup A_{3}|}$} \\\hline
$q^{m-1}-\sum_{i=1}^{3}q^{|A_{i}|-1}+q^{|A_{2} \cap A_{3}|-1}$ &
$q^{m-|(A_{1} \cap A_{2}) \cup (A_{1} \cap A_{3})|}-q^{m-|A_{1}|}-q^{m-|(A_{1} \cap A_{2}) \cup (A_{1} \cap A_{3})\cup (A_{2} \cap A_{3})|}+q^{m-|A_{1}\cup(A_{2}\cap A_{3})|}$  \\\hline
$q^{m-1}-\sum_{i=1}^{3}q^{|A_{i}|-1}+q^{|A_{1} \cap A_{3}|-1}$ &
$q^{m-|(A_{1} \cap A_{2}) \cup (A_{2} \cap A_{3})|}-q^{m-|A_{2}|}-q^{m-|(A_{1} \cap A_{2}) \cup (A_{1} \cap A_{3})\cup (A_{2} \cap A_{3})|}+q^{m-|A_{2}\cup(A_{1}\cap A_{3})|}$  \\\hline
$q^{m-1}-\sum_{i=1}^{3}q^{|A_{i}|-1}+q^{|A_{1} \cap A_{2}|-1}$ &
$q^{m-|(A_{1} \cap A_{3}) \cup (A_{2} \cap A_{3})|}-q^{m-|A_{3}|}-q^{m-|(A_{1} \cap A_{2}) \cup (A_{1} \cap A_{3})\cup (A_{2} \cap A_{3})|}+q^{m-|A_{3}\cup(A_{1}\cap A_{2})|}$   \\\hline
$q^{m-1}-\sum_{i=1}^{3}q^{|A_{i}|-1}+q^{|A_{1} \cap A_{3}|-1}+q^{|A_{2} \cap A_{3}|-1}$ &
$q^{m-|A_{1} \cap A_{2}|}-q^{m-|(A_{1} \cap A_{2}) \cup(A_{1} \cap A_{3})|}-q^{m-|(A_{1} \cap A_{2}) \cup(A_{2} \cap A_{3})|}+q^{m-|(A_{1} \cap A_{2}) \cup (A_{1} \cap A_{3})\cup (A_{2} \cap A_{3})|}$  \\\hline
$q^{m-1}-\sum_{i=1}^{3}q^{|A_{i}|-1}+q^{|A_{1} \cap A_{2}|-1}+q^{|A_{2} \cap A_{3}|-1}$ &
$q^{m-|A_{1} \cap A_{3}|}-q^{m-|(A_{1} \cap A_{2}) \cup(A_{1} \cap A_{3})|}-q^{m-|(A_{1} \cap A_{3}) \cup(A_{2} \cap A_{3})|}+q^{m-|(A_{1} \cap A_{2}) \cup (A_{1} \cap A_{3})\cup (A_{2} \cap A_{3})|}$  \\\hline
$q^{m-1}-\sum_{i=1}^{3}q^{|A_{i}|-1}+q^{|A_{1} \cap A_{2}|-1}+q^{|A_{1} \cap A_{3}|-1}$ &
$q^{m-|A_{2} \cap A_{3}|}-q^{m-|(A_{1} \cap A_{2}) \cup(A_{2} \cap A_{3})|}-q^{m-|(A_{1} \cap A_{3}) \cup(A_{2} \cap A_{3})|}+q^{m-|(A_{1} \cap A_{2}) \cup (A_{1} \cap A_{3})\cup (A_{2} \cap A_{3})|}$   \\\hline
$q^{m-1}-\sum_{i=1}^{3}q^{|A_{i}|-1}+\sum_{1\leq i < j \leq 3}q^{|A_{i} \cap A_{j}|-1}$ &
\tabincell{l}{$q^{m-|A_{1} \cap A_{2} \cap A_{3}|}-q^{m-|A_{1} \cap A_{2}|}-q^{m-|A_{1} \cap A_{3}|}-q^{m-|A_{2} \cap A_{3}|}+q^{m-|(A_{1} \cap A_{2}) \cup (A_{1} \cap A_{3})|}$\\$+q^{m-|(A_{1} \cap A_{2}) \cup (A_{2} \cap A_{3})|}+q^{m-|(A_{1} \cap A_{3}) \cup (A_{2} \cap A_{3})|}-q^{m-|(A_{1} \cap A_{2}) \cup (A_{1} \cap A_{3})\cup (A_{2} \cap A_{3})|}$}   \\\hline
$q^{m-1}-q^{-1}|\Delta|$&
$q^{m}-q^{m-|A_{1} \cap A_{2} \cap A_{3}|}$  \\\hline
\end{tabular}
\end{table}

\begin{thm} \label{optimalcode-anti-sc-h=3}
Let $\Delta$ be a simplicial complex of $\fqm$ with the support $\mathcal{A}=\{A_{1},A_{2},A_{3}\}$, where $1 \leq |A_{1}| \leq |A_{2}| \leq |A_{3}|<m$.
Assume that $A_{i}\backslash (\cup_{1\leq j \leq 3, j\ne i}A_{j}) \ne \emptyset$ for any $1\leq i \leq 3$, and
$q^{m} > \sum_{1\leq i\leq 3} q^{|A_{i}|}$. Let $T=\sum_{1\leq i\leq 3}q^{|A_{i}|-1}$. Then $\C_{\overline{\Delta}^{c}}$ defined by \eqref{CD} is an at most $19$-weight $[(q^m-|\Delta|)/(q-1),m,q^{m-1}-T]$ linear code with weight distribution in Table \ref{wd-h=3}, where $|\Delta|=\sum_{i=1}^{3}q^{|A_{i}|}-\sum_{1\leq i < j \leq 3}q^{|A_{i} \cap A_{j}|}+q^{|A_{1} \cap A_{2} \cap A_{3}|}$.
Moreover, we have the followings:
\begin{enumerate}
\item [1)] $\C_{\overline{\Delta}^{c}}$ is a Griesmer code if and only if $|A_{i}\cap A_{j}|=0$ for $1\leq i<j \leq 3$ and at most $q-1$ of $|A_{i}|$'s are the same (which always holds for $q>3$).
\item [2)] $\C_{\overline{\Delta}^{c}}$ is a near Griesmer code if one of the followings holds: i) $|A_{i}\cap A_{j}|=1$ for only one element $(i,j)$ in the set $\{(i,j): 1\leq i<j \leq 3\}$ and $|A_{i}\cap A_{j}|=0$ for the other two $(i,j)$'s, and at most $q-1$ of $|A_{i}|$'s are the same; ii) $q=3$, $|A_{i}\cap A_{j}|=0$ for $1\leq i<j \leq 3$, and $|A_{1}|=|A_{2}|=|A_{3}|$; and iii) $q=2$, $|A_{i}\cap A_{j}|=0$ for $1\leq i<j \leq 3$, and $|A_{1}|=|A_{2}|<|A_{3}|-1$ or $|A_{1}|\leq |A_{2}|=|A_{3}|$.
\item [3)] $\C_{\overline{\Delta}^{c}}$ is distance-optimal if $(q-1)(v(T)+1)+\ell(T)-1>\sum_{1\leq i < j \leq 3}q^{|A_{i} \cap A_{j}|}-q^{|A_{1} \cap A_{2} \cap A_{3}|}$.
\end{enumerate}
\end{thm}

\begin{proof}
The parameters of $\C_{\overline{\Delta}^{c}}$ follow directly from 1) of Theorem \ref{optimalcode-anti-sc-pro}.
By using the formula given in 4) of Theorem \ref{optimalcode-anti-sc-pro}, the weight distribution of $\C_{\overline{\Delta}^{c}}$ can be obtained as $R$ runs through all subsets of $\Omega=\{S: S \subseteq \mathcal{A}, S\ne \emptyset \}=\{\,\{A_{1}\},\{A_{2}\},\{A_{3}\},\{A_{1},A_{2}\},\{A_{1},A_{3}\},\{A_{2},A_{3}\},\{A_{1},A_{2},A_{3}\}\,\}$. Then the weight distribution of $\C_{\overline{\Delta}^{c}}$ is given as in Table \ref{wd-h=3} by a routine computation step by step. Next we discuss the optimality of $\C_{\overline{\Delta}^{c}}$.

Note that the length of $\C_{\overline{\Delta}^{c}}$ is $n=\frac{1}{q-1}(q^m-|\Delta|)$ and the minimum distance of $\C_{\overline{\Delta}^{c}}$ is $d=q^{m-1}-\sum_{1\leq i\leq 3}q^{|A_{i}|-1}$.  By Lemma \ref{g(m,d)}, we have
\begin{equation} \label{eq-cor4-1}
g(m,d)= \frac{1}{q-1}(q^m-\sum_{1\leq i\leq 3}q^{|A_{i}|}+\ell(T)-1),
\end{equation}
where $T=\sum_{1\leq i\leq 3}q^{|A_{i}|-1}$. It can be readily verified that when $q>3$, $\ell(T)=3$; when $q=3$, $\ell(T)=1$ if $|A_{1}|=|A_{2}|=|A_{3}|$ and $\ell(T)=3$ otherwise; and when $q=2$, it gives
\[\ell(T)=\left\{\begin{array}{ll}
1,     &   \mbox{ if $|A_{1}|=|A_{2}|=|A_{3}|-1$};\\
2,     &   \mbox{ if $|A_{1}|=|A_{2}|<|A_{3}|-1$ or $|A_{1}|\leq |A_{2}|=|A_{3}|$};\\
3,  &   \mbox{ otherwise}.
\end{array} \right.\]

According to 2) of Theorem \ref{optimalcode-anti-sc-pro}, it proves part 1). If one of the conditions i), ii) and iii) in part 2) is satisfied, one has $n-g(m,d)=1$ by \eqref{eq-cor4-1} and the values of $n$ and $\ell(T)$, which implies that $\C_{\overline{\Delta}^{c}}$ is a near Griesmer code. This proves part 2). Part 3) follows directly from 3) of Theorem \ref{optimalcode-anti-sc-pro}. This completes the proof.
\end{proof}

\begin{remark}
To the best of our knowledge, our result in Theorem \ref{optimalcode-anti-sc-h=3} is the first to construct linear codes by employing simplicial complexes with more than two maximal elements in the literature.
\end{remark}

\begin{remark}
Note that $v(T) \geq |A_{1}|-1$ and $1\leq \ell(T)\leq 3$ in Theorem \ref{optimalcode-anti-sc-h=3} by the definition. Therefore, the sufficient condition $(q-1)(v(T)+1)+\ell(T)-1>\sum_{1\leq i < j \leq 3}q^{|A_{i} \cap A_{j}|}-q^{|A_{1} \cap A_{2} \cap A_{3}|}$ in 3) of Theorem \ref{optimalcode-anti-sc-h=3} such that $\C_{\overline{\Delta}^{c}}$  is distance-optimal can be easily satisfied if $|A_{1}|$ is large enough and $|A_{i} \cap A_{j}|$'s for $1\leq i < j \leq 3$ are small enough. Hence many distance-optimal linear codes can be obtained from Theorem \ref{optimalcode-anti-sc-h=3}.
\end{remark}

\begin{example}
Let $q=3$, $m=5$, $A_{1}=\{1\}$, $A_{2}=\{2\}$ and $A_{3}=\{3\}$. Magma experiments show that $\C_{\overline{\Delta}^{c}}$ is a $[118,5,78]$ linear code over ${\mathbb F}_{3}$ with weight enumerator $1+72z^{78}+108z^{79}+54z^{80}+8z^{81}$, which is consistent with our result in Theorem \ref{optimalcode-anti-sc-h=3}. This code is a near Griesmer code by the Griesmer bound and is optimal due to \cite{GMB}.
\end{example}

\begin{example}
Let $q=3$, $m=5$, $A_{1}=\{1,2\}$, $A_{2}=\{1,3\}$ and $A_{3}=\{4\}$. Magma experiments show that $\C_{\overline{\Delta}^{c}}$ is a $[113,5,74]$ linear code over ${\mathbb F}_{3}$ with weight enumerator $1+24z^{74}+120z^{75}+54z^{76}+24z^{77}+12z^{78}+6z^{80}+2z^{81}$, which is consistent with our result in Theorem \ref{optimalcode-anti-sc-h=3}. This code is a near Griesmer code by the Griesmer bound and is almost optimal due to \cite{GMB}.
\end{example}

\begin{example}   
Let $q=2$, $m=12$, $A_{1}=\{1,2,3,4,5\}$, $A_{2}=\{1,2,6,7,8,9\}$ and $A_{3}=\{1,3,4,6,7,8,10,11\}$. Magma experiments show that $\C_{\overline{\Delta}^{c}}$ is a $[3770,12,1872]$ linear code over $\ftwo$ with weight enumerator $1+6z^{1872}+24z^{1874}+48z^{1876}+96z^{1878}+112z^{1880}+224z^{1882}+672z^{1884}+2048z^{1885}+672z^{1886}+6z^{1888}+112z^{1896}
+6z^{1904}+48z^{1908}+6z^{1920}+2z^{2000}+8z^{2002}+2z^{2016}+2z^{2032}+z^{2048}$, which is consistent with our result in Theorem \ref{optimalcode-anti-sc-h=3}. This is actually a $19$-weight code, which means that every possible Hamming weight of $\C_{\overline{\Delta}^{c}}$ in Table \ref{wd-h=3} can be achieved. In addition, by the Griesmer bound, an upper bound of the minimum distance for a $[3770,12]$ code over $\fq$ is given by $1884$ which is quite close to $1872$ due to $1872/1884\approx 0.99363$.
\end{example}

\section{A comparison to the previous works}\label{sec6}

Note that infinite families of (distance-)optimal codes over $\fq$ can be produced from our codes $\C_{\overline{\Delta}^{c}}$ and the parameters of $\C_{\overline{\Delta}^{c}}$ over $\fq$ are very flexible. Now we compare our codes $\C_{\overline{\Delta}^{c}}$ with the known (distance-)optimal linear codes in the previous works. At first, it can be readily checked that our codes $\C_{\overline{\Delta}^{c}}$ constructed in this paper have different parameters to the optimal linear codes in \cite{HDWZ,HZHZ,HJKN,LXSM,LXYQ,PYLY,SMLX,WYHJ,ZXWY}. Next, we make
the comparison of our optimal codes to the known ones with similar parameters (see their parameters in Table \ref{optimal-table}) as follows:
\begin{itemize}
  \item Our codes $\C_{\overline{\Delta}^{c}}$ over $\fq$ are new for $q>2$ compared to the binary linear codes of \cite{JYHLL} since our results extend those of \cite{JYHLL} from $\ftwo$ to $\fq$. Moreover, when $q=2$, compared to \cite{JYHLL}, we further investigate the weight distribution of $\C_{\overline{\Delta}^{c}}$ for the general case and take a more in-depth discussion on the case of simplicial complexes with three maximal elements.
  \item Our codes $\C_{\overline{\Delta}^{c}}$ over $\fq$ are different from the well-known Solomon-Stiffler codes \cite{GSJS} over $\fq$ except when $|A_{i}\cap A_{j}|=0$ for any $1\leq i<j \leq h$ and at most $q-1$ of $|A_{i}|$'s are the same, in which case $\C_{\overline{\Delta}^{c}}$ is a Griesmer code (see Theorem \ref{optimalcode-anti-sc-pro}). Viewed from the construction approach, the Solomon-Stiffler codes are constructed from the complement of the union of disjoint projective subspaces of $\PG(m-1,q)$ where at most $q-1$ of these subspaces have the same dimension, while our codes $\C_{\overline{\Delta}^{c}}$ are constructed from $\overline{\Delta}^{c}$ which can be viewed as the complement of the union of projective subspaces of $\PG(m-1,q)$ without the restriction ``disjoint'' and that on the number of subspaces with the same dimension. This is their crucial difference.
  \item Our codes $\C_{\overline{\Delta}^{c}}$ over $\fq$ are different from those of \cite{HLZWT} in general even though our codes have the same parameters as those of \cite{HLZWT} in some special cases. Note that the codes of
      \cite[Theorem 1]{HLZWT} were constructed from the complement of the union of some distinct subfields ${\mathbb F}_{q^{r_{i}}}$ of $\fqm$, where $r_{i}$'s divide $m$ and any two subfields ${\mathbb F}_{q^{r_{i1}}}$ and ${\mathbb F}_{q^{r_{i2}}}$ must intersect in the subfield ${\mathbb F}_{q^{r}}$ of $\fqm$ with $r=\gcd(r_{i1},r_{i2})$. However, our $|A_{i}|$'s in $\overline{\Delta}^{c}$ are not necessary to divide $m$ and any two distinct projective subspaces of $\PG(m-1,q)$ with fixed dimensions in $\overline{\Delta}^{c}$ can be disjoint or intersect in some projective subspace with any dimension. Therefore, the parameters of our codes $\C_{\overline{\Delta}^{c}}$ are more flexible than those of \cite[Theorem 1]{HLZWT} in the projective case even though their parameters have the same form (see Table \ref{optimal-table}). When $|A_{1}|=\cdots=|A_{h}|=\ell$ and $\ell$ divides $m$, our codes $\C_{\overline{\Delta}^{c}}$ have the same parameters as those of \cite[Theorem 3]{HLZWT} in the projective case. When $h=2$, our codes $\C_{\overline{\Delta}^{c}}$ (see Theorem \ref{optimalcode-anti-sc-h=2}) have more flexible parameters than those of \cite[Theorem 4]{HLZWT} and \cite[Theorem 5]{HLZWT} in the projective case even though their parameters have the same form in this case; in addition, they have different weight distributions.
  \item  Our codes $\C_{\overline{\Delta}^{c}}$ over $\fq$ are different from these (distance-)optimal codes in \cite{CHJY,HJKWY,LYDCTC,WYLCXF,WYZXYQ} since the parameters of our codes are more flexible. However, our codes have the same parameters with these codes in some special cases. When $h=1$ and $|A_{1}|=1$, $h=2$ and $|A_{1}|=|A_{2}|=1$, and $(q,h)=(2,2)$, $|A_{1}|=2$, $|A_{2}|=m-1$ and $|A_{1}\cap A_{2}|=1$,
      our codes $\C_{\overline{\Delta}^{c}}$ have the same parameters with those of \cite[Thm. 18]{LYDCTC},
      \cite[Thm. 19]{LYDCTC} and \cite[Thm. 22]{LYDCTC}, respectively.
      When $(q,h)=(2,1)$ and $m$ is even, our code $\C_{\overline{\Delta}^{c}}$ has the same parameters with that of
      \cite[Thm. 5.2]{WYLCXF} and it reduces to the code of \cite[Thm. 4.3]{WYZXYQ} if $|A_{1}|=m-1$.
      When $(q,h)=(2,2)$, $|A_{1}|=1$ and $|A_{2}|=m-1$, our code $\C_{\overline{\Delta}^{c}}$ has the same parameters as the optimal binary codes in \cite{CHJY,LYDCTC}. When $(q,h)=(2,2)$, $|A_{1}|=1$ (resp. $(q,h)=(2,2)$, $|A_{1}|=2$ and $|A_{1}\cap A_{2}|=1$), our code $\C_{\overline{\Delta}^{c}}$ has the same parameters with that of
      \cite[Thm. 6.1]{HJKWY} (resp. \cite[Thm. 6.2]{HJKWY}).
\end{itemize}

\begin{table}[!htb]\footnotesize
\caption{The parameters $[n,k,d]$ of some (distance-)optimal linear codes over $\fq$}  \label{optimal-table}
\centering
\begin{tabular}{|c|c|c|c|c|c|}
\hline No.&$q$ & $n$&    $k$        & $d$      &   Remarks                  \\ \hline\hline
  1 &$2$ &$2^m-1$&      $m+1$         & $2^{m-1}-1$       &  $m>1$, \cite{CHJY,LYDCTC}  \\\hline

  2 &$2$ &$2^m-2$&      $m+1$         & $2^{m-1}-2$                  & $m\geq 3$,     \cite[Thm. 22]{LYDCTC} \\\hline
  3 &$2$ &$2^{2m-1}$       & $2m$     & $2^{2m-2}$       &    \cite[Thm. 4.3]{WYZXYQ}  \\\hline
  4 &$2$ &$2^n-2^m-1$&      $n$         & $2^{n-1}-2^{m-1}-1$       & $2<m<n$,  \cite[Thm. 6.1]{HJKWY}  \\\hline
  5 &$2$ &$2^n-2^m-2$&      $n$         & $2^{n-1}-2^{m-1}-2$       & $1<m<n-1$,  \cite[Thm. 6.2]{HJKWY} \\\hline
  6 &$2$ &$2^{2m}-2^{|A|+|B|}$       & $2m$     & $2^{2m-1}-2^{|A|+|B|-1}$       &    \cite[Thm. 5.2]{WYLCXF}  \\\hline
  7 &$2$ & $2^n-\sum_{\emptyset \not=S\subseteq \mathcal{F}}(-1)^{|S|-1}2^{|\cap S|}$
         &  $n$ & $2^{n-1}-\sum_{i=1}^{s}2^{|A_{i}|-1}$ &    \cite{JYHLL}  \\\hline
  8 &any &$(q^m-q)/(q-1)-1$&      $m$         & $q^{m-1}-1$       &    \cite[Thm. 18]{LYDCTC} \\\hline
  9 &any &$(q^m-2q+1)/(q-1)$&      $m$         & $q^{m-1}-2$       & $m\geq 2$, \cite[Thm. 19]{LYDCTC} \\\hline
  10 &any &$q^m-\sum_{\emptyset \not=S\subseteq \Upsilon}(-1)^{|S|-1}q^{r_{S}}$&      $m$
     & $(q-1)(q^{m-1}-\sum_{i=1}^{h}q^{r_{i}-1})$       &   \cite[Thm. 1]{HLZWT}    \\\hline
  11 &any &$q^m-(h+1)q^r$&      $m$         & $(q-1)(q^{m-1}-(h+1)q^{r-1})$       & $h+1\leq q$,   \cite[Thm. 2]{HLZWT}  \\\hline
  12 &any &$q^m-(h+1)q^r$&      $m$         & $(q-1)q^{m-1}-hq^{r}$       & $h+1>q$,    \cite[Thm. 2]{HLZWT}    \\\hline
  13 &any &$q^m-hq^r+h-1$&      $m$         & $(q-1)(q^{m-1}-hq^{r-1})$       &    \cite[Thm. 3]{HLZWT}    \\\hline
  14 &any &$(q^{m}-q^{r})(q^{k}-q^{s})$&      $m+k$         & $(q-1)(q^{m+k-1}-q^{m+s-1}-q^{k+r-1})$       &
       \cite[Thm. 4]{HLZWT}    \\\hline
  15 &any &$q^{m+k}-q^{m}-q^{k}+1$&      $m+k$         & $(q-1)(q^{m+k-1}-q^{m-1}-q^{k-1})$  &   \cite[Thm. 5]{HLZWT} \\\hline
  16 &any & $(q^m-1-\sum_{i=1}^{h}(q^{u_{i}}-1))/(q-1)$& $m$  & $q^{m-1}-\sum_{i=1}^{h}q^{u_{i}-1}$  & Solomon-Stiffler code \cite{GSJS}\\ \hline
  17 &any &$(q^m-\sum_{\emptyset \not=S\subseteq \mathcal{A}}(-1)^{|S|-1}q^{|\cap S|})/(q-1)$& $m$ &
     $q^{m-1}-\sum_{1\leq i\leq h} q^{|A_{i}|-1}$       &    This paper   \\\hline
\end{tabular}
\end{table}
\section{Concluding remarks} \label{sec5}

The main contributions of this paper are summarized as follows:
\begin{itemize}
  \item We constructed a large family of projective linear codes $\C_{\overline{\Delta}^{c}}$ over $\fq$ from the general simplicial complexes $\Delta$ of $\fq^m$ via the defining-set construction. This totally extends the results of \cite{JYHLL} from $\ftwo$ to $\fq$. To the best of our knowledge, this paper is the first to study linear codes over $\fq$ constructed from the general simplicial complexes of $\fq^m$ for a prime power $q>2$.
  \item The parameters and weight distribution of $\C_{\overline{\Delta}^{c}}$ were completely determined (see Theorem \ref{optimalcode-anti-sc-pro}) in this paper. Thus this paper also determines the weight distribution of the binary codes constructed from the general simplicial complexes of $\ftwo^m$ in \cite[Theorem IV.6]{JYHLL}, in which the weight  distribution of the binary codes were studied only for the case of simplicial complexes of $\ftwo^m$ with two maximal elements. Moreover, as a byproduct, the weight distributions of the Solomon-Stiffler codes are determined in Corollary \ref{optimalcode-anti-sc-pro-case1} for the case that the corresponding subspaces in $\fq^m$ of the projective subspaces $U_{i}$ are spanned by some subsets of a certain basis of $\fq^m$.
  \item By using the Griesmer bound, we gave a necessary and sufficient condition such that  $\C_{\overline{\Delta}^{c}}$ is a Griesmer code and a sufficient condition such that $\C_{\overline{\Delta}^{c}}$ is distance-optimal. In addition, we also presented a necessary and sufficient condition for $\C_{\overline{\Delta}^{c}}$ to be a near Griesmer code in a special case. This shows that many infinite families of (distance-)optimal linear codes can be produced from our construction.
  \item By studying the cases of the simplicial complexes $\Delta$ with one, two and three maximal elements respectively, we   presented the parameters and weight distributions of these codes $\C_{\overline{\Delta}^{c}}$ more explicitly, and  derived infinite families of optimal linear codes with few weights over $\fq$ including Griesmer codes, near Griesmer codes and distance-optimal codes (see Theorems \ref{optimalcode-anti-sc-h=1}, \ref{optimalcode-anti-sc-h=2} and \ref{optimalcode-anti-sc-h=3}).
  \item Most notably, the definition of simplicial complexes of $\fq^m$ of this paper and the main technique used to prove the results in this paper can also be employed to extend the previous results of \cite{CHJY,LXSM,WYLY,WYLCXF,WYZXYQ,ZXWY}, where simplicial complexes of $\ftwo^m$ with one and two maximal elements were utilized to constructed optimal few-weight binary or quaternary linear codes.
\end{itemize}

The construction of optimal linear codes is an interesting problem. The reader is cordially invited to construct more new optimal linear codes.


\end{document}